\def\arxiv{1}
\ifnum\arxiv=0
    \documentclass[sigconf]{acmart}
    \usepackage{natbib}
\else
    \documentclass[10pt]{article}
    \usepackage{fullpage}
    \usepackage{amsthm}
    \newtheorem{theorem}{Theorem}
    \newtheorem{lemma}[theorem]{Lemma}
    
    \theoremstyle{definition}
    
    \usepackage[square,sort,comma,numbers]{natbib}
\fi

\def\BibTeX{{\rm B\kern-.05em{\sc i\kern-.025em b}\kern-.08emT\kern-.1667em\lower.7ex\hbox{E}\kern-.125emX}}

\usepackage{nicefrac}
\usepackage{siunitx}
\usepackage{array,framed}
\usepackage{booktabs}
\usepackage{
  color,
  float,
  epsfig,
  wrapfig,
  graphics,
  graphicx,
  subcaption
}
\usepackage{textcomp,amssymb,amsthm}
\usepackage{setspace}
\usepackage{latexsym,fancyhdr,url}
\usepackage{enumerate}
\usepackage{algorithm2e}
\usepackage{algpseudocode}
\usepackage{graphics}
\usepackage{xparse} \usepackage{xspace}
\usepackage{multirow}
\usepackage{csvsimple}
\usepackage{balance}

\usepackage{
  tikz,
  pgfplots,
  pgfplotstable
}
\usepackage{hyperref}

\usetikzlibrary{
  shapes.geometric,
  arrows,
  external,
  pgfplots.groupplots,
  matrix
}

\pgfplotsset{compat=1.9}

\usepackage{mathtools}

\DeclareMathAlphabet{\mathcal}{OMS}{cmsy}{m}{n}

\clearpage{}
\newcommand{\ikosview}{\mathcal{V}_{m, n}}
\newcommand{\altshuffler}[1]{\mathcal{A}_{#1}}
\newcommand{\altshufflertwo}[1]{\mathcal{A}_{#1}^2}

\newcommand{\alttwo}{\altshufflertwo{}}
\newcommand{\alttwop}{\altshufflertwo{\pi}}

\newcommand{\G}{\mathbb{G}}
\newcommand{\rv}[1]{\mathsf{#1}}
\newcommand{\tup}[1]{\vec{#1}}
\newcommand{\dom}[1]{\mathbb{#1}}
\newcommand{\mech}[1]{\mathcal{#1}}
\newcommand{\TV}{\mathfrak{T}}
\newcommand{\shuffleCols}{\texttt{ShuffleCols}}
\renewcommand{\TV}{\mathrm{TV}}
\newcommand{\mat}[1]{\mathsf{#1}}
\newcommand{\sk}{\mathbf{sk}}
\newcommand{\pk}{\mathbf{pk}}
\newcommand{\Committee}{C}
\newcommand{\Enc}{\texttt{Enc}}
\newcommand{\committeesize}{n_{\texttt{dec}}}
\newcommand{\polylog}{\texttt{polylog}}
\newcommand{\shufflecommitteesize}{n_{\texttt{shuf}}}
\newcommand{\shuffleclient}{\mathbf{sc}}
\newcommand{\server}{\mathcal{S}}
\newcommand{\Bin}{\mathrm{Bin}}

\newcommand{\borja}[1]{\textcolor{red}{Borja: #1}}
\newcommand{\james}[1]{\textcolor{blue}{James: #1}}
\newcommand{\adria}[1]{\textcolor{purple}{Adria: #1}}
\renewcommand{\borja}[1]{}
\renewcommand{\james}[1]{}
\renewcommand{\adria}[1]{}
\clearpage{}
\usepackage{cleveref}

\DeclareGraphicsExtensions{.png,.PNG,.pdf,.PDF,.jpg,.mps,.jpeg,.jbig2,.jb2,.JPG,.JPEG,.JBIG2,.JB2}

\begin{document}

\title{Amplification by Shuffling without Shuffling}
\date{\today}
\author{{\rm Borja Balle} \\ {\rm Google DeepMind} \and {\rm James Bell} \\ {\rm Google} \and {\rm Adrià Gascón} \\ {\rm Google} }

\maketitle

\begin{abstract}
Motivated by recent developments in the shuffle model of differential privacy, we propose a new \emph{approximate} shuffling functionality called Alternating Shuffle, and provide a protocol implementing alternating shuffling in a single-server threat model where the adversary observes all communication.
Unlike previous shuffling protocols in this threat model, the per-client communication of our protocol only grows sub-linearly in the number of clients.
Moreover, we study the concrete efficiency of our protocol and show it can improve per-client communication by one or more orders of magnitude with respect to previous (approximate) shuffling protocols.
We also show a differential privacy amplification result for alternating shuffling analogous to the one for uniform shuffling, and demonstrate that shuffling-based protocols for secure summation based a construction of Ishai et al.~\cite{IKOS} remain secure under the Alternating Shuffle.
In the process we also develop a protocol for exact shuffling in single-server threat model with amortized logarithmic communication per-client which might be of independent interest.

\end{abstract}

\section{Introduction}

The shuffle model of differential privacy (DP) has emerged in recent years as an appealing intermediate between the classical central and local models which enables accurate private computations in distributed settings without transmitting ``plain-text'' data to a trusted aggregator \cite{DBLP:conf/sosp/BittauEMMRLRKTS17,DBLP:conf/eurocrypt/CheuSUZZ19,DBLP:conf/soda/ErlingssonFMRTT19}.
The key building block of the shuffle model is a \emph{trusted shuffler}: a black-box primitive that receives as input a collection of messages submitted by individuals and returns a random permutation of those messages, thus obfuscating their origin and establishing an important contrast with local model protocols where an adversary can track messages back to the individual from whom they originated.
Theoretical protocols leveraging one or more of these primitives have been proposed for a wide range of differentially private computations, including boolean summation \cite{DBLP:conf/eurocrypt/CheuSUZZ19}, real summation \cite{DBLP:conf/crypto/BalleBGN19,DBLP:conf/ccs/BalleBGN20,DBLP:conf/eurocrypt/GhaziMPV20}, histograms \cite{DBLP:conf/icits/BalcerC20,DBLP:conf/eurocrypt/GhaziG0PV21,DBLP:conf/soda/BalcerCJM21,DBLP:journals/corr/abs-2104-02739} and machine learning \cite{DBLP:journals/corr/abs-2001-03618}.
The accuracy of these protocols significantly surpasses the best possible protocols in the local model, and often matches the accuracy of central model mechanisms.
Together with further theoretical work establishing lower bounds and impossibility results (see \cite{DBLP:journals/corr/abs-2107-11839} and references therein), available protocols illustrate the power and practical promise of the shuffle model, and at the same time highlight important separations between the local, shuffle and central models.

A fine-grained separation also arises within the class of shuffle model protocols when one considers the number of messages each user sends through the trusted shuffler.
For example, protocols where each user sends a single message through the shuffler are strictly less powerful that protocols where each user is allowed to send three messages through the shuffler \cite{DBLP:conf/ccs/BalleBGN20,DBLP:conf/eurocrypt/GhaziMPV20}.
This observation has fuelled research into the trade-offs between communication and accuracy in the shuffle model, with the number and size of messages sent per user being used as the main proxy for communication complexity \cite{DBLP:conf/ccs/BalleBGN20,DBLP:conf/icits/GhaziGKMPV20,DBLP:conf/icml/Ghazi0MPS21,DBLP:journals/corr/abs-2104-02739}.
These works also explore, often implicitly, small but important variations in the threat model of multi-message shuffle model protocols ranging from a relying on single trusted shuffler admitting an arbitrary number of messages from each user, to having access to a fixed number of independent trusted shufflers each admitting a single message from each user.
While both communication complexity and threat modelling assumptions are extremely relevant factors for practical applications, the trade-offs involved in translating these results into implementations instantiating concrete trusted shuffler primitives have received significantly less attention.

The seminal work of Bittau et al.\ \cite{DBLP:conf/sosp/BittauEMMRLRKTS17} on the Encode, Shuffle and Analyze framework proposed to instantiate a trusted shuffler using a trusted execution environment (e.g.\ Intel SGX) hosted by an honest-but-curious server who receives the output of the shuffler to perform some analysis.
On the other hand, while formalizing the shuffle model, Cheu et al.\ \cite{DBLP:conf/eurocrypt/CheuSUZZ19} suggest mixnets \cite{chaum1981untraceable} as a potential method for realizing the shuffling functionality.
Bell et al.\ \cite{DBLP:conf/ccs/BellBGL020} propose a secure aggregation protocol for vector summation that can be used to instantiate a cryptographically secure shuffler with linear communication per user.
Their protocol works in the single-server model, where each user can securely communicate with an honest-but-curious powerful server, and provides strong cryptographic guarantees in the presence of an adversary simultaneously corrupting the server and a small fraction of the users.
This is a natural threat model for distributed data analysis tasks where a powerful but untrusted server is tasked with analyzing data distributed across a large number of less powerful devices, and consequently has received significant attention from both theoretical and practical perspectives \cite{DBLP:conf/ccs/BonawitzIKMMPRS17,DBLP:journals/corr/abs-2207-05047, acorn, flamingo, microfedfl, melisetal, DBLP:conf/ndss/ShiCRCS11}.

Inspired by similar questions about the concrete practical efficiency of shuffle model protocols, recent works investigate the use of \emph{approximate} shufflers which might admit more efficient implementations.
Along these lines, Gordon et al.\ \cite{Gordonetal} propose the notion of \emph{differentially oblivious} (DO) shufflers to capture a relaxation of the perfect shuffling assumption where, instead of asking that every possible permutation is equally likely, one requires that distributions over permutations obtained from inputs differing by a single transposition are indistinguishable (w.r.t.\ the same notion of indistinguishability used in differential privacy).
Gordon et al.\ show that DO shufflers can be used to replace perfect shufflers in a restricted class of shuffle model protocols (i.e.\ those based \emph{amplification} of local DP guarantees \cite{DBLP:conf/soda/ErlingssonFMRTT19,BalleBGN20,DBLP:conf/focs/FeldmanMT21}) with only a small degradation of the final privacy guarantees.
They also propose a DO shuffler based on onion routing that works in the single-server threat model with logarithmic communication per client but poor concrete efficiency incurred by the requirement to perform a large number of rounds of communication.
B\"{u}nz et al.\ \cite{DBLP:journals/iacr/BunzHMS21} give an alternative implementation of a DO shuffler, however their protocol requires a trusted third party in a setup phase and is not dropout resistant. Zhou et al.\ \cite{DBLP:journals/iacr/ZhouSCM22} investigate the compositional properties of differential obliviousness, give a tighter analysis of the privacy amplification properties of DO shufflers, and provide examples on how to instantiate some multi-message shuffle model protocols not directly based on amplification using DO shufflers.
Unfortunately, there exist important classes of shuffle model protocols (e.g.\ optimal multi-message protocols for real summation \cite{DBLP:conf/ccs/BalleBGN20, DBLP:conf/eurocrypt/GhaziMPV20}) which so far have not been shown to be realizable using DO shufflers.

In this work we identify \emph{alternating shuffling}, the first approximate shuffling primitive capable of overcoming the limitations of DO shuffling. Broadly speaking, a single round of alternating shuffling approximates the shuffling of $n$ messages by first arranging them into $\sqrt{n}$ rows and columns, and then performing a shuffle across rows followed by a shuffle across columns.
Our work provides an in-depth analysis of the properties of this approximate shuffler with regards to the implementation of protocols in the shuffle model of DP. This shows that Differential Obliviousness is not necessary for amplification.
Furthermore, we propose an implementation of the alternating shuffling primitive that is cryptographically secure in the honest-but-curious single-server threat model, uses sublinear communication per user, requires a small number of rounds, and is resilient to both dropouts and a small fraction of users colluding with the server. We also state some open problems in the analysis of this shuffler and provide and analyze a protocol for single server shuffling (which we call the amortized shuffler) that we built as a partial result, and may be of independent interest.

The rest of the paper is organized as follows.
\Cref{sec:setup} introduces the main functionality and threat model used throughout the paper.
In \Cref{sec:overview} we provide a high-level overview of our contributions, focusing on the key challenges our techniques allow us to address.
Our first set of results revolves around the use of our alternating shuffling functionality to implement distributed DP protocols, including protocols that rely on amplification of local DP guarantees as well as simulating a secure summation protocol using the ideas from \cite{IKOS} -- this is done in \Cref{sec:properties}, where we also survey some interesting open problems.
Finally, \Cref{sec:implementations} presents secure instantiations of approximate shuffling functionalities, together with an analysis of their concrete efficiency.

\section{Setup}\label{sec:setup}

\subsection{Preliminaries}

\paragraph{Differential privacy.}
Two random variables $A$ and $B$ over the same space $Z$ are said to be $(\epsilon, \delta)$-indistinguishable, denoted by $A \simeq_{\epsilon,\delta} B$, if for every $E \subseteq Z$ we have
\begin{align*}
    \max\{\Pr[A \in E] - e^{\epsilon} \Pr[B \in E], \Pr[B \in E] - e^{\epsilon} \Pr[A \in E]\} \leq \delta \enspace. 
\end{align*}
If $p$ and $q$ denote the probability distributions of $A$ and $B$ respectively, we also write $p \simeq_{\epsilon,\delta} q$.
We ignore $\delta$ in this notation whenever it is zero.

A \emph{local randomizer} is a randomized map $R : X \to Y$ from the space of inputs $X$ to the space of messages $Y$. We say that $R$ is $\epsilon_0$-LDP if we have $R(x) \simeq_{\epsilon_0} R(x')$ for all $x, x' \in X$.
A randomized mechanism $M : X^n \to Z$ is $(\epsilon,\delta)$-DP if for any datasets $D, D' \in X^n$ differing in a single element we have $M(D) \simeq_{\epsilon,\delta} M(D')$.

We recall two basic results about amplification of DP guarantees by sampling and shuffling.

\begin{lemma}[\cite{DBLP:conf/nips/BalleBG18}]\label{lem:sampling}
Let $p, q, r$ be distributions such that $p \simeq_{\epsilon_0,\delta} q$ and $p \simeq_{\epsilon_0,\delta} r$. For any $\gamma \in [0,1]$ the mixtures $p' = (1-\gamma) r + \gamma p$ and $q' = (1-\gamma) r + \gamma q$ satisfy $p' \simeq_{\epsilon,\gamma \delta} q'$ with $\epsilon = \epsilon_{\mathrm{sampling}}(\epsilon_0, \gamma) := \log(1 + \gamma (e^{\epsilon_0} - 1))$.
\end{lemma}

\begin{theorem}[\cite{DBLP:journals/corr/abs-2208-04591}]\label{thm:uniform-shuffling}
Let $R$ be an $\epsilon_0$-LDP local randomizer. Suppose $n$ and $\delta$ are such that $\epsilon_0 \leq \log\left(\frac{n}{8 \log(2/\delta)} - 1\right)$. Then the protocol obtained by uniformly shuffling $n$ copies of $R$ is $(\epsilon, \delta)$-DP with
\begin{align*}
    \epsilon
    =
    \epsilon_{\mathrm{clones}}(\epsilon_0, \delta, n)
    :=
    \log\left(1 + (e^{\epsilon_0}-1) \left(\sqrt{\frac{32 \log(4/\delta)}{(e^{\epsilon_0}+1) n}} + \frac{4}{n} \right)\right) \enspace.
\end{align*}
\end{theorem}

Note that for any $\epsilon_0 = \log O(n)$ for which the theorem applies we get $\epsilon_{\mathrm{clones}}(\epsilon_0, \delta, n) = O(1)$.

\paragraph{Differentially oblivious shuffling}
A mapping $S: Y^n \to Y^n$ applying a random permutation to its inputs is $(\epsilon,\delta)$-differentially oblivious (DO) if for any two inputs $y, y'$ differing in a transposition (i.e.\ $y_i = y_j'$, $y_j = y_i'$ for some $i \neq j$, and $y_k = y_k'$ for $k \neq i,j$) we have $S(y) \simeq_{\epsilon,\delta} S(y')$. This definition was introduced in \cite{Gordonetal} -- see the reference for a more general version of the definition involving corrupted clients.
Building on this concept, and leveraging privacy amplification by uniform shuffling, Zhou et al.~\cite{DBLP:journals/iacr/ZhouSCM22} show that a $(\epsilon_1,\delta_1)$-DO shuffler can amplify an $\epsilon_0$-LDP local randomizer to provide an $(\epsilon+\epsilon_1, \delta+\delta_1)$-DP protocol with $\epsilon = O((e^{\epsilon_0/2}-e^{-\epsilon_0/2}) \sqrt{\log(1/\delta)/n})$. The result extends to a setting with $t$ corrupted clients by replacing $n$ with $n-t$ and ensuring the DP shuffler maintains its guarantees under that number of corruptions.

\paragraph{ElGamal cryptosystem}
In our implementations we will utilize ElGamal encryption in a group $G$ of order $p$ generated by $g$ within which the decisional Diffie-Hellman assumption holds. For assessing communication costs we assume 256 bit elements. A private key is a random integer $\sk\in \mathbb{Z}_p$ and the corresponding public key is $g^s$. The encryption of a message $m$ (encoded as a group element) is given by $(mg^{\sk r},g^r)$ where $r$ is uniformly random.
Given a ciphertext $(c_1,c_2)$ encrypted with public key $\pk=g^\sk$, anyone with $\pk$ and an integer $a$ can change the key needed to decrypt to $a+\sk$, even without knowing $\sk$, by replacing the ciphertext with $(c_1c_2^a,c_2)$.
This \emph{key homomorphism} property also makes it possible for any $t$ members of a committee holding $t$-out-of-$n$ Shamir shares of a secret key $\sk$ to help another party decrypt a ciphertext, without learning anything themselves and without revealing the key giving the other party the key. We will use this to enable a committee of clients to hold the secret key material and only help the server to decrypt certain encrypted values.
A further advantage of ElGamal encryption that we make heavy use of is the existence of an efficient zero-knowledge proof of correct behaviour for the task of taking a list of ciphertexts, permuting them and re-encrypting them. This proof is due to Bayer and Groth~\cite{BayerG12}, it adds communication overhead that is small compared to the list of ciphertexts and requires on the order of ten exponentiations per ciphertext. This will allow clients to shuffle values for the server whilst preventing any malicious behaviour.

\subsection{Setting \& Threat Model}

In our setting, a single server, denoted $\server$ orchestrates a computation with a large number $n$ of clients, arranged in a star network topology, with $\server$ at the center.
In our protocols, $\server$ acts as a relay between
clients: $\server$ enables key exchange between clients
to establish an authenticated secure channel, and routes subsequent encrypted messages.

As in cross-device Federated Learning~\cite{advances}, we expect clients to be resource-constrained, and possibly have limited connectivity. Therefore, practical protocols must be robust to a reasonable fraction of dropouts, which we denote by $\alpha$. Our security guarantees do not rely on the number of dropouts.

\paragraph{Functionality.}
We assume that every client $i$ holds a private input $x_i$
from a large domain. The functionality implemented by our protocols applies a random permutation (abstractly denoted $\texttt{shuffle}$ below, more details later), and gives the result to the server, as long as the fraction of dropped out clients stays below $\alpha$. More precisely, our shuffling functionalities
are parameterized by a set of dropout clients $D\subseteq [n]$ and the dropout robustness parameter $\alpha$, and the server obtains an output $\Pi_{D, \alpha}(x_1, \ldots, x_n)$ defined as follows:
$$
\Pi_{D, \alpha}(x_1, \ldots, x_n) =
\left\{
	\begin{array}{ll}
		\texttt{shuffle}\big(\{x_i\}_{i\not\in D}\big) & \mbox{if } |D| < \alpha n, \\
		\bot & \mbox{otherwise.}
	\end{array}
\right.
$$

On the other hand, clients involved in the protocol get no output.

\paragraph{Threat model.} Our protocols withstand 
the following adversaries (informally stated). We assume static corruptions, i.e. malicious clients are set before the protocol execution starts, and do not change throughout.
\begin{itemize}
    \item Coalition of $n-1$ malicious clients: A set of up to $n-1$ clients controlled by a polynomial-time adversary and behaving arbitrarily will not learn anything about honest inputs.
    \item Coalition of a semi-honest Server and up to $\gamma$n malicious clients: An adversary simultaneously controlling up to a $\gamma$-fraction of the clients, and observing the server's protocol transcript will not learn anything about the inputs of honest clients.
\end{itemize}
Furthermore, so long as the server and all but a fraction $\alpha$ of the clients follow the protocol honestly (and do not drop out) the server will receive a multiset with one input from each client who did not drop before a set point. 
Our formal security proofs are in the ideal vs.\ real paradigm, using standard simulation-based security in multi-party computation~\cite{goldreich, howtosimulate}.

An important aspect of out threat model is that, since the server enables client-to-client communication, 
it observes the communication pattern. This is 
an important observation when designing secure shuffling functionalities. In particular, the work of Gordon et al.~\cite{Gordonetal}, as well as other solutions based on onion routing, assumes a weaker adversary that does not have access to communication patterns.

\section{Overview of Contributions}
\label{sec:overview}

\paragraph{Towards efficient shuffling in the single-server threat model.}
We are working in a threat model in which there is a single semi-honest server mediating all communications between clients.
Thus, in order to shuffle with a permutation not known to the server (and potentially colluding malicious clients), it would be convenient to have the (honest) clients do the permuting.
If there was one client known to be honest they could just collect encrypted inputs from every client (through the server), and return a shuffle of those inputs to the server.
Since a priori there is no way to identify such a client, instead we could select a number of clients (independent of the total number of clients participating) to do the shuffling, and have them each shuffle in turn -- this approach is reminiscent of the system proposed by Chaum~\cite{chaum1981untraceable}.
To prevent a malicious client acting as a shuffler from replacing some or all of the ciphertexts, a secure instantiation of this scheme requires that client provides a zero-knowledge proof that they have done this correctly.
Fortunately, this is feasible with only a small constant factor in computational overhead thanks to the specialised proof of Bayer and Groth~\cite{BayerG12} and the re-encryption properties of the ElGamal cryptosystem.
If ElGamal keys are generated amongst a single decryption committee, this approach yields a shuffling protocol with $O(n(\sigma+\log(n)))$ total communication.
Details of this protocol are spelled out in Figure~\ref{Protocol:CTShuffle}.

\paragraph{Achieving linear communication with amortized shuffling.}
To reduce the heaviest computation any one client does in the above protocol from $O(n(\sigma+\log(n)))$ to $O(n)$, we propose a novel setup that allows many committees of the clients to generate independent sharings of the same (fresh) ElGamal secret key.
This setup is described in Figure~\ref{Protocol:KeyAgreement}.
Putting this together with the previous protocol gives our Amortized Shuffling protocol (see Figure~\ref{Protocol:Amortized}), which has \emph{amortized} constant communication per client.
Whilst the total communication and number of rounds of this protocol are good, we would like sublinear communication for \emph{all} clients.

\paragraph{Achieving sub-linear communication with alternating shuffling.}
When clients are in charge of ciphertext shuffling, a key idea to reduce the communication required by such clients is to shuffle only a \emph{subset} of ciphertexts each time, iterating until the overall collection of ciphertexts are sufficiently shuffled.
Ideas in the mixnet literature suggest a protocol along these lines could be implemented with $O(\mathrm{polylog}(n))$ communication per client, however the number of rounds this protocol would require would be in the hundreds so we are not satisfied with that approach.
Instead we take an intermediate path that can be described as arranging the ciphertexts in a square matrix, and then shuffling first each row independently, following by transposing the matrix and iterating a number of times.
We call this (approximate) shuffling functionality the \emph{Alternating Shuffler} (denoted by $\altshuffler{\pi, r}^\ell$), which is formally described in Figure~\ref{fig:alternatin-shuffler}. Unless stated otherwise, throughout the paper we assume that $h = w = \sqrt{n}$ -- this gives a protocol with $O(\sqrt{n})$ per-iteration communication per client.
The protocol in Figure~\ref{Protocol:Alternating} provides a secure implementation of such functionality.

\begin{figure*}[ht]
\ifnum\arxiv=0
$
\begin{pmatrix}
    y_{1}\\
     \\ \\
    \vdots \\
     \\ \\
    y_{n}
\end{pmatrix}\xrightarrow[\text{\emph{publicly}}]{\ \text{Shuffle input}\ }
\begin{pmatrix}
    y_{\pi ( 1)}\\
     \\ \\
    \vdots \\
     \\ \\
    y_{\pi ( 1)}
\end{pmatrix}\xrightarrow{\ \text{Rearrange}\ }
\begin{pmatrix}
m_{1,1} & \cdots  & m_{1,w}\\
m_{2,1} & \cdots  & m_{2,w}\\
 &  & \\
m_{h,1} & \cdots  & m_{h,w}
\end{pmatrix}\xrightarrow[\text{and transpose}]{\ \text{Shuffle each row \emph{privately}}\ }
\begin{pmatrix}
    m_{1,\ \pi _{1}( 1)} & \cdots  & m_{h,\ \pi _{h}( 1)}\\
    m_{1,\ \pi _{1}( 2)} & \cdots  & m_{h,\ \pi _{h}( 2)}\\
     &  & \\
    m_{1,\ \pi _{1}( w)} & \cdots  & m_{h,\ \pi _{h}( w)}
\end{pmatrix}\text{~}$\\[.1ex]
$\hspace{8cm}\underbrace{\hspace{.40\textwidth}}_{\text{Repeat $\ell$ times}}$
\else
\begin{center}
       \includegraphics[width=0.9\textwidth]{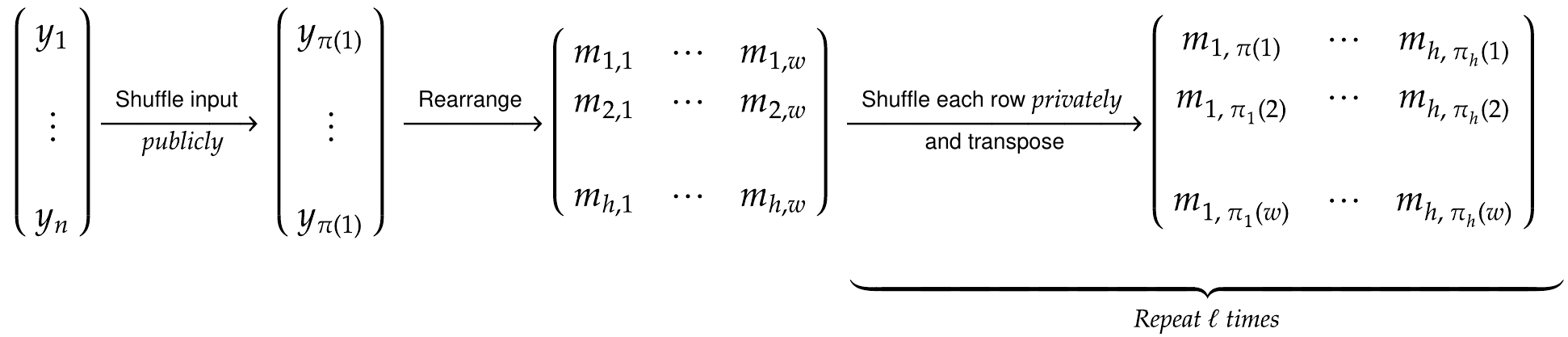}
\end{center}

\fi \caption{\small An $\ell$-round \emph{Alternating Shuffler}, 
denoted $\altshuffler{\pi, r}^\ell$, 
consists of the following steps:
(1) arrange the input randomly into a $h\times w$ grid $\mat{M} = (m_{i,j})_{i,j}$ using a public permutation $\pi$;
(2) shuffle the $h$ rows of $\mat{M}$ independently using $h$ private permutations $\pi_j: [w]\to [w]$ sampled from private randomness $r$;
(3) transpose $\mat{M}$; and
(4) repeat steps (2-3) $\ell$ times in total.
We shall drop $\pi$ from the notation when there is no public permutation, and also drop $r$ when it is not relevant or understood from the context.}
\label{fig:alternatin-shuffler}
\end{figure*}

The remaining question is how many iterations of row shuffling are required by such protocol. Håstad showed in~\cite[Theorem 3.6]{Hastad} that $O(1+\sigma/\log(n))$ iterations would suffice to provide a shuffle that is (approximately) uniform to within statistical distance $2^{-\sigma}$. However this would still require a fairly large number of rounds to achieve near-perfect shuffling.
Thus we ask the question: do two or three iterations of alternating shuffling suffice to provide enough privacy for differentially private data analysis protocols?

\paragraph{Properties of alternating shuffling.}
In Section~\ref{sec:weak} we prove that shuffling the rows twice suffices to provide a weak form of amplification by shuffling, in which the resulting $\epsilon$ scales with $e^{5\epsilon_0/2}$, where $\epsilon_0$ is the privacy parameter of the local randomizer. Whilst we do not believe the constant $5/2$ to be optimal, we prove in Section~\ref{sec:no-strong-amp} that it cannot be improved to less than $1$. The corresponding constant in the case of amplification by uniform shuffling is known to be $1/2$ (see Theorem~\ref{thm:uniform-shuffling}).
Furthermore, in Section~\ref{sec:ikos} we prove that shuffling rows twice is sufficient to securely implement the summation via shuffling protocol of Ishai et. al.~\cite{IKOS} (which we refer to as the IKOS protocol).
This enables us to implement the best known protocols for DP summation in the shuffle model \cite{BalleBGN20,DBLP:conf/eurocrypt/GhaziMPV20} using only an approximate shuffler.
We note that this result does not rely on amplification, and is not known to be possible using DO shufflers.

We leave open whether shuffling the rows three times gives strong amplification (i.e.\ with $1/2$ in the exponent) and it is also open whether it is sufficiently DO to imply strong amplification. We also note that randomly \emph{publicly} shuffling the inputs before applying two row shuffles might suffice to provide strong amplification, despite the fact that we show in Section~\ref{sec:not-do} that this is not DO with any non-trivial parameters.

\subsection{Related Work and Concrete Efficiency}
As mentioned in the introduction, several recent works 
have tackled secure single-server shuffling.
On the more theoretical front, B\"unz et al.~\cite{DBLP:journals/iacr/BunzHMS21}
show that {\em non-interactive} oblivious shuffling with sublinear server computation is possible,
and follows from standard assumptions in bilinear groups. The main drawback is that the construction requires a trusted setup.
Alon et al.~\cite{alonetal} provide a general solution to secure computation in our setting, and show that 
any efficient function with $O(n \polylog(n))$-sized output can be computed
while ensuring that each user's communication and computation costs are $\polylog(n)$. Moreover, 
the number of rounds of the protocol is also $\polylog(n)$.
This general construction requires Fully Homomorphic Encryption (FHE). The authors also show that for
simpler (but useful) functionalities such as summation and shuffling FHE is not needed.
On the other hand, the works of Mohavedi et
al.~\cite{Mohavedietal}, Bell et al.~\cite{DBLP:conf/ccs/BellBGL020}, and Gordon et al.~\cite{Gordonetal} are more focused
on practicality.
We discuss these results next, and provide a concrete analysis, both asymptotically
and in terms of concrete efficiency, in Table~\ref{tab:comparison}. For a discussion of classical approaches such
as those based on mix-nets and dining cryptographers networks we refer the reader to~\cite{Mohavedietal}.

The approach by Mohavedi et al.~\cite{Mohavedietal} is based on a combination of techniques from multi-party computation.
Roughly speaking, the protocol boils down to 
a cohort-based secure evaluation of a (probabilistic) sorting network $\mathcal{C}$
of $O(\log n)$ depth and $O(n\log n)$ gates, to shuffle $n$ user-provided values. 
First, clients agree on $n$ cohorts of size $O(\log n)$, and then each cohort evaluates 
a few gates securely, secret sharing the output of the evaluation to the corresponding cohort for the next gate.
Since each client belongs to 
a logarithmic number of cohorts, and each cohort evaluates 
a logarithmic number of gates, appropriate use of
efficient verifiable secret sharing and distributed cohort formation 
techniques results in a protocol with polylogarithmic work per client.
The fact that the depth of $\mathcal{C}$ is logarithmic leads to $O(\log n)$ rounds.

Bell et al.~\cite{DBLP:conf/ccs/BellBGL020} propose a constant-round vector summation protocol,
along with a reduction from shuffling to vector summation via linear sketching, and in particular using the Inverse Bloom Look-up Table~\cite{IBLT} data structure. This reduction results in a sketch of size $O(n)$, with $\log(n)$ bits per entry, as input for the aggregation, and thus $O(n)$ per-client work. On the other
hand this protocol requires just $4$ rounds.

As mentioned in the introduction, Gordon et al.~\cite{Gordonetal} propose a protocol for differentially Oblivious shuffling.
An important remark is that their threat model is weaker than ours, in that they assume that users can communicate independently of the server, 
and thus the communication pattern between users is not revealed to the adversary. This enables an onion routing based approach were clients use layered 
encryption to enforce that their message travels through a random sequence of users before reaching the server. 
The length of the sequence corresponds to the round complexity of the protocol. For uniform shuffling this is required to be quite long,
resulting in impractical costs. Gordon et al's observation is that for DO-shuffling is not the case.
The high-level observation is that a transposition of two messages can be realized when those messages are held by 
two honest clients in a given round of the protocol. For uniform shuffling one needs to realize up to $n$
such transpositions with large enough probability, while a single one is enough for differential obliviousness.
This results in a round complexity for Gordon et al.'s approach that is
independent of $n$, and depending only on the DP parameters $\epsilon, \delta$. Another important observation is that 
this onion routing based approach is not robust to dropouts, unlike our approach, Bell et al.'s and Mohavedi et al.'s. We note also that Bell at al. and Mohavedi et al. are only robust to a fixed fraction of malicious clients, whereas our protocol is robust against $n-1$ malicious clients.

As shown in Table~\ref{tab:comparison}, our amortized shuffler protocol has average costs that are 
either independent of the number of users $n$ (round complexity and computation),
or logarithmic (communication). This improves on all previous approaches. 
However, worst-case cost for a small number of parties is linear in $n$.
As we will see next our constants are small, and even our $O(n)$ worst-case costs compare favorably with previous
approaches. Finally, our 2-round alternating shuffling protocol achieves worst-case costs 
and round complexity independent of $n$ which, combined with the result from Håstad~\cite{Hastad},
yields a protocol with sublinear costs in $n$ and number of rounds that improves with $n$.
We now turn our attention to the second half of Table~\ref{tab:comparison}, to focus on concrete efficiency.

\paragraph{Concrete efficiency.}
Table~\ref{tab:comparison} shows two configurations of interest chosen as a comparison point with previous works. In 
both cases the number of clients is $33000$, and in one case no dropouts are expected ($\alpha = 0$)
while in the other protocols are expected to be robust to up to $1$ in $10$ clients to drop out.
The maximum fraction of malicious clients in $\gamma = 1/10$ and $1/20$, respectively.
The table shows how even the worst-case costs for the amortized shuffler are competitive with previous works,
while the remaining costs (average costs and number of rounds) are significantly better with the exception of the number of rounds in Bell et al. However, it is important to 
remark that in each round of Bell et al. all clients need to do work and send data to the server, 
while in most of our protocol's round only one client has to do work.
This means that whilst the protocol of Bell et al. would complete much faster if heavily parallelized, we could beat Bell et al. on wall clock time as well if the throughput to a single server during their masked input submission were a sufficient bottleneck.

We note that our amortized shuffling protocol is sufficiently efficient to be practical.
The average cost is independent of $n$ and fairly cheap in both communication and computation
(the main expense being the distributed decryptions).
The main restriction on when this protocol is practical is that the clients must be willing to pay $n$ times the average cost with probability about $1/n$, or alternatively the cost of $O(n)$ runs all at once if the client will run the protocol many times.
The alternating shuffle protocol is more amenable to production environments as the number of rounds is not too high and the communication cost per-client is $\tilde{O}(\sqrt{n})$ with a reasonable constant.
More details about the concrete efficiency evaluation of our protocols are provided in Section~\ref{sec:feasibility}.

\begin{table*}
\scriptsize
\begin{tabular}{lcccccc}
\hline
  &
  \multicolumn{1}{c}{\begin{tabular}[c]{@{}c@{}}Gordon et al.\\\cite{Gordonetal}\end{tabular}} &
  \multicolumn{1}{c}{\begin{tabular}[c]{@{}c@{}}Mohavedi et al.\\\cite{Mohavedietal}\end{tabular}} &
  \multicolumn{1}{c}{\begin{tabular}[c]{@{}c@{}}Bell et al.\\\cite{DBLP:conf/ccs/BellBGL020} \end{tabular}} &
  \multicolumn{1}{c}{\begin{tabular}[c]{@{}c@{}}Amortized\\ (Thm.~\ref{thm:amortized}) \end{tabular}} &
  \multicolumn{1}{c}{\begin{tabular}[c]{@{}c@{}}$2$-round\\ Alternating \\(Thm.~\ref{thm:alternatingasymptotics})\end{tabular}} &
  \multicolumn{1}{c}{\begin{tabular}[c]{@{}c@{}}Alternating \\(Thm.~\ref{thm:alternatingasymptotics}~\cite{Hastad})\end{tabular}}\\ \hline
Computation      & $O(\delta^2\log n)$ & $\polylog(n)$ & $O(n\log n)$              & $O(n)~/~O(\sigma)$ & $O(\sqrt{n})$ & $O(\sqrt{n})$ \\
Communication    & $O(\delta^2\log n)$ & $\polylog(n)$ & $O(n\log n    )$ & $O(n)~/~O(\sigma+\log n)$        & $O(\sqrt{n})$ & $O(\sqrt{n})$ \\
Round complexity &
  $O(\delta)$  &
  $O(\log n)$ &
  $O(1)$ &
  $O(\sigma)$ &
  $O(\sigma)$ &
  $O(\sigma + \sigma^2/\log n)$ \\
Functionality    & DO-Shuffling  & Shuffling  & Shuffling                    & Shuffling                 & \multicolumn{1}{c}{\begin{tabular}[c]{@{}c@{}}2-round\\ Alternating\end{tabular}} & Shuffling \\[1ex] \hline
\multicolumn{7}{l}{\begin{tabular}[c]{@{}c@{}}Concrete Efficiency $(\sigma=\log(1/\delta) = 13, n = 33,000, \alpha = 0, \gamma = 1/3, \textrm{128-bit inputs})$\end{tabular}}\\ \hline
~~~Communication    & 182-390KB & 128MB & 13MB & 4MB / 3KB & 35KB / 5KB & -- \\
~~~Number of rounds & 70-103 & 500 & 4 & 14 & 35 & -- \\ \hline
\multicolumn{7}{l}{\begin{tabular}[c]{@{}c@{}}Concrete Efficiency $(\sigma=\log(1/\delta) = 40, n = 33,000, \alpha = 1/10, \gamma = 1/20, \textrm{128-bit inputs})$\end{tabular}}\\ \hline
~~~Communication    & N/A & 128MB & 13MB & 4MB / 7KB & 27KB / 11KB & -- \\
~~~Number of rounds & N/A & 500 & 4 & 34 / 21 & 69 / 41 & -- \\ \hline
\end{tabular}
\vspace{5mm}
\caption{\small
Per-client costs of different protocols along with some properties and round complexity.
Entries written as A/B denote maximum per-user cost (A) and average per-user cost (B) when specifying communication, and worst and best case depending on how the dropouts fall when specifying number of rounds.
We expect the lower number to be the most relevant in practice.
The concrete numbers for Gordon et al. and Mohavedi et al.\ come from \cite{Gordonetal}; they only provide a range including $n=33,000$, hence the ranges.
Gordon et al.\ cannot tolerate dropouts, hence the N/A entries.
The protocol on the last column relies on an asymptotic result with no constants provided, hence the --.
Note shuffling with $\sigma$ as a security parameter provides the same guarantees as DO-shuffling with $\delta=2^{-\sigma}$, the reverse is not true.}\label{tab:comparison}
\end{table*}

\section{Properties of Alternating Shuffling}\label{sec:properties}

\subsection{Private Data Analysis via Anonymity}

The shuffle model of DP enables distributed data analysis protocols where $n$ users each holding a private input $x_i \in X$ collaborate with an analyzer to privately compute a statistic on the dataset $D = (x_1, \ldots, x_n)$ without sending the plain-text data to the analyzer.
Instead, the model relies on (one or more) trusted shuffling primitives capable of ``anonymously'' sending messages from the users to the analyzer, plus a randomization primitive run by the users before submitting their data to the shuffler.
Note that here we assume the shufflers operate as a perfectly secure black-box (i.e.\ we focus on the functionality rather than the protocol used to implement it). In this setting, one is interested in the privacy provided to the individual users by the view available to the analyzer (and, by post-processing, to anyone who gets access to the final result they release after the analysis).
In this context, we recall two general families of shuffling-based private data analysis protocols -- in the rest of this section we investigate how these paradigms extend from uniform shuffling to alternating shuffling.

\paragraph{Privacy amplification.}
In the privacy amplification paradigm, one considers single-message shuffling protocols where each user $i$ applies an $\epsilon_0$-LDP local randomizer $R$ to their data $x_i$ to obtain a message $y_i = R(x_i)$.
The analyzer then receives $S(y_1, \ldots, y_n)$, the result of applying a uniform random permutation to the users' messages, and is tasked with producing the result of the analysis.
Privacy of the protocol is based on analyzing the effect of changing one user's data on the view of the analyzer after the shuffling, e.g.\ using \Cref{thm:uniform-shuffling}.

\paragraph{Privacy via secure summation.}
This paradigm relies on using the IKOS \cite{IKOS} construction for implementing secure summation via shuffling in order to provide a distributed implementation of the standard output perturbation mechanism for summation in the central model of DP.
In this case, users essentially add to their input an $n$th fraction of the total noise required to privatize the sum of their messages, then split the noisy input into multiple additive shares inside a large enough group, and finally send each of the shares to the aggregator via a separate shuffler.
When using $m$ shares, this results in an $m$-message shuffling protocol. See \cite{DBLP:conf/ccs/BalleBGN20,GhaziMPV20} for further details.

\subsection{Alternating Shuffler Gives ``Weak'' Amplification}\label{sec:weak}

The amplification by \emph{uniform} shuffling result given in Theorem~\ref{thm:uniform-shuffling} has the form $\epsilon = O(e^{\epsilon_0/2})$ in terms of its dependence on $\epsilon_0$. The $1/2$ constant in the exponent is tight, and important to achieve good privacy-utility trade-offs in amplification-based protocols -- previous amplification by shuffling results had worst constants in the exponent. We call this level of amplification of local DP guarantees \emph{strong}. Here we prove that alternating shuffling provides \emph{weak} amplification, in the sense that the dependence on $n$ is the same as in Theorem~\ref{thm:uniform-shuffling}, but the constant in the exponent is worse. We will later show some worsening is in fact unavoidable.

\begin{theorem}\label{thm:weak-amplification}
Suppose $n$ is sufficiently large and $\epsilon_0 = O(1)$.
If the local randomizer $R$ is $\epsilon_0$-LDP, then applying the alternating shuffler $\alttwo$ to the outputs of $R$ a yields a protocol satisfying $(\epsilon_A, \delta_A)$-DP with
$\epsilon_A = O\left(\frac{e^{5\epsilon_0/2} \log(1/\delta)}{\sqrt{n}}\right)$ and $\delta_A = \left(1 + \frac{\sqrt{n} e^{2 \epsilon_0}}{\sqrt{n} + e^{2 \epsilon_0}}\right) \delta$.
\end{theorem}

At a high level, the proof works by bounding the privacy loss incurred by releasing each individual column and then applying a composition analysis to bound the total privacy loss of the protocol.
Informally speaking, each column is the result of applying uniform shuffling to a database with $\sqrt{n}$ individuals, which has a privacy loss of order $n^{-1/4}$ (cf.\ Theorem~\ref{thm:uniform-shuffling}).
This privacy is further amplified by realizing that only one of these columns will contain the user that differs between the two databases.
Since the probability the differing user lands in a particular columns is $n^{-1/2}$, we obtain that each columns incurs a privacy loss of the order $n^{-1/4} \cdot n^{-1/2}$ (cf.\ Lemma~\ref{lem:sampling}).
Applying the advanced composition theorem \cite{DBLP:journals/tit/KairouzOV17} to the privacy loss incurred by $\sqrt{n}$ columns then yields a total privacy loss of the order $n^{-1/4} \cdot n^{-1/2} \cdot n^{1/4} = n^{-1/2}$ as claimed in the theorem.

We note that our proof does not strive to optimize the constant $5/2$ in the power of $e^{\epsilon_0}$.
In fact, it is possible to slightly improve this constant at the cost of a more cumbersome analysis and worse constants in the big-$O$ by using a non-homogeneous composition argument and a column-dependent $\gamma$ (see e.g.\ \cite{DBLP:conf/nips/BalleKMTT20}). We defer the details to future work.

\paragraph{Corrupt Clients.}
The analysis above assumes none of the users involved in the protocol collude with the server.
Nonetheless, the analysis can be extended to the case where a small fraction of the total number of users collude with the server by paying a small degradation in the final privacy guarantees.
However, achieving this requires an additional (public) permutation of the users to be applied before the protocol starts -- this ensures that not too many of the colluding users take positions in the row containing the user whose data is being attacked.

\begin{theorem}\label{thm:weak-amplification-corrupted}
Consider the setting of \Cref{thm:weak-amplification} where a public random permutation is applied to the users before the start of the protocol and where up to $\gamma n$ users collude with the server.
Then the protocol satisfies $(\epsilon_A, \delta_A)$-DP with
$\epsilon_A = O\left(\frac{e^{5\epsilon_0/2} \log(1/\delta)}{(1-\gamma) \sqrt{n}}\right)$ and $\delta_A = \left(3 + \frac{\sqrt{n} e^{2 \epsilon_0}}{\sqrt{n} + e^{2 \epsilon_0}}\right) \delta$.
\end{theorem}

\subsection{Alternating Shuffler Doesn't Give Strong Amplification}\label{sec:no-strong-amp}

Here we will give an example of a pair of inputs and a local randomizer, with local epsilon $0.51\log(n)$, on which amplification by $\alttwo$ doesn't give a constant $\epsilon$ for any non-trivial $\delta$. That is amplification by $\alttwo$ is asymptotically less powerful than amplification by shuffling.

\begin{theorem}\label{thm:no-strong}
There exists a family of pairs of databases $D_0,D_1$ and local randomizers $R$ which are locally $(0.5+o(1))\log(n)$-DP, but for which amplification by $\alttwo$ fails to provide $(O(1),1-\Omega(1))$-DP.
\end{theorem}

\subsection{IKOS  in the Alternating Shuffler model}\label{sec:ikos}
In this Section we show that 
the the summation via shuffling protocol of Ishai et al.~\cite{IKOS}
(which we refer to as the IKOS protocol)
is secure in the $2$-round alternating shuffler model. In the original protocol, 
each client splits their input into $m$
additive shares, and sends them into the (uniform) shuffler. The server then adds all resulting shuffled shares to get the resulting sum. Ishai et al. showed that $m = O(\log n)$
shares per client are enough for security.
This result was later improved by  Ghazi et al.~\cite{GhaziMPV20} and Balle et al.~\cite{BalleBGN20}, who showed that in fact $m= O(1)$ suffices. The main result in this section is that the same $m = O(1)$ bound
applies to instantiations of IKOS with a 
$2$-round alternating suffler, instead of a uniform shuffler. 

We start by defining formally the local processing in the IKOS protocol, i.e. how client $i$
obtains the messages to be shuffled given their input $x_i$. Specifically,
for any $\tup{x} = (x_1,\ldots,x_n) \in \dom{G}^n$, we define the random variables $\mech{R}_{m}(x_i) := (\rv{Y}^{(1)}_i, \ldots, \rv{Y}^{(m)}_i)$, $i \in [n]$, obtained by splitting each input into $m$ additive shares.

We identify the $m$-parallel IKOS protocol over $\G$
with the randomized map $\mech{V}_{m,n} : \dom{G}^n \to (\dom{G}^n)^{m}$
defined next, and
corresponding to the view of the aggregator in an $m$-message protocol in the {\em $2$-round alternating} shuffle model with randomizer $\mech{R}_{m}$:

\begin{align*}
\ikosview(\vec{x}) = \Big(\altshufflertwo{\pi, r_1}&\big(\rv{Y}^{(1)}_1, \ldots, \rv{Y}^{(1)}_n\big), \ldots, \altshufflertwo{\pi, r_m}\big(\rv{Y}^{(m)}_1, \ldots, \rv{Y}^{(m)}_n\big)\Big)
\end{align*}

Note that this model assumes $m$ {\em single-message} shufflers $\altshufflertwo{\pi, r_1}$, $\ldots$, $\altshufflertwo{\pi, r_n}$ that share the {\em same public randomness}, and each with {\em their own secret internal randomness} $r_i$.

We next show that the IKOS protocol in the alternating shuffler setting is secure with $m = O(1)$, for sufficiently large $n$, thus recovering the results from Ghazi et al.~\cite{GhaziMPV20} and Balle et al.~\cite{BalleBGN20} in the uniform shuffler model.
The idea of the proof is to view the construction of 
$\ikosview(\vec{x})$
as a series of applications of the IKOS protocol with 
uniform permutations,
but over subsets of the input of size $\sqrt{n}$. These subsets corresponds to 
rows (or columns) of the $h\times w$
matrix in the internal state of the alternating shuffler as shown in Figure~\ref{fig:alternatin-shuffler} (we assume that $h=w = \sqrt{n}$).
The proof thus applies the result of Balle et al.~\cite{BalleBGN20}
$1+\sqrt{n}$ times (to all rows of the internal state matrix, and then to a column, of each of the $m$ single-message alternating shufflers that constitute our model). Note we use that the result in Balle et al. doesn't require all the messages to be shuffled together there can be $m$ different shuffles that each client puts one message into and their result still holds. We have not substantially optimized the constants in this or the following derived theorem.

\begin{theorem}\label{theorem:ikos-all-honest-clients}
Let $n \geq 361$ and $m \geq 3$.
The protocol $\ikosview$ provides worst-case statistical security with parameter 
$$\sigma = (m-2)\left(\frac{1}{2}\log_2(n) - \log_2(e)\right) - \log_2(q) - 2 \enspace.$$
Therefore the required number of messages per client is 
$$m = O\left(1+\frac{\sigma + \log_2(q)}{\log_2(n)}\right) \enspace.$$
\end{theorem}

The requirement on $n$, and the $\log_2(n)$ term due to the union bound across IKOS instances we don't believe to be necessary. In fact, we suspect that a more direct proof than the one above,
adapting the ideas from  Balle et al.~\cite{BalleBGN20} instead of using their result as a black box probably works. We leave a detailed proof of this approach for further work.

We now extend the above result to the setting where up to a $\gamma$-fraction of the clients might be 
dishonest and collude with the server. The idea for this extension is simple: 
the public randomness $\pi$ 
induces a high-probability lower bound
on the number of honest clients 
in each row of the input matrices
to be shuffled.
This is enough to port the argument in the proof of the previous theorem to the setting with corrupted clients.

\begin{theorem}[IKOS with Corrupted Inputs]\label{thm:ikos-corrupted}
Let $(1-\gamma)\sqrt{n} - (\sigma + \log n)^{1/2}n^{1/4} \geq 19$ and $m \geq 3$.
The protocol $\ikosview$ is robust to up to $\gamma n$ corrupt clients, and provides worst-case statistical security for all parameters $\sigma$ such that
$$\sigma \leq 
(m-2)\log_2\left(\frac{(1-\gamma)n^{1/2}-(\sigma + \log n)n^{1/4}}{e}\right) - \log_2 q - 3 \enspace.$$
Therefore, for large enough $n$ and $q\in \texttt{poly}(n)$, the required number of messages per client $m$ is $O(1)$.\
\end{theorem}

Equipped with a communication efficient 
exact summation protocol, we now turn our 
attention to the problem of diferentially private real summation.
We use the reduction from secure summation to differentially private real summation
by Balle et al.~\cite{BalleBGN20} (Theorem 5.2) to obtain a 
$O(1)$ error protocol, thus matching the error of  central model in the (alternating) shuffler model. The 
basic ideas of the reduction are to
(i) apply an appropriate quantization scheme
of the real-valued input to balance quantization error with DP noising error, and 
(ii) simulate noise addition is a distributed way by relying on infinite divisibility properties of {\em discrete} random variables, i.e. a geometric random variable can be expressed as a sum of negative binomial random variables. 

Given real inputs in $[0,1]$ we can round them to the nearest multiple of $1/\sqrt{n}$, multiply by $\sqrt{n}$ to get an integer and do the addition modulo $2n\sqrt{n}$ (which is then the value of $q$). This gives the following result.

\begin{theorem}[Constant Error DP Summation]
There exists an $(\epsilon, \delta)$-DP
protocol in the multi-message alternating shuffler model for real summation
with MSE $O(1/\epsilon^2)$ and  
$O(1+\log(1/\delta)/\log(n))$ messages,
each of $O(\log n)$ bits in size.
\end{theorem}

\subsection{Alternating Shuffler with Public Randomness is not DO}\label{sec:not-do}

The lower bound in Section~\ref{sec:no-strong-amp} relies on a very structured counter example (there is a row of almost all ones and no other ones). Therefore merely applying a public uniformly random shuffle to the input, before the alternating shuffler is applied, is enough to break that counter example. Of course the structured input may retain its structure through the shuffle but the probability of that (at least for the specific structure used in the counterexample) is negligible for moderately large $n$.
We do not know whether $\alttwop$ provides strong amplification. One approach to proving that it is would be to prove that it is DO with sufficiently small parameters. However we are able to show that that approach won't work.

\begin{theorem}\label{thm:no-do}
$\alttwop$ is not $DO$ for any $\epsilon$ and any
$\delta < 1 - \frac{2}{n+1}$.
\end{theorem}

\section{Implementations of (Approximate) Shuffling}
\label{sec:implementations}

As discussion in Section~\ref{sec:overview}, existing algorithms for shuffling and DO approximations of shuffling either require linear communication from each client or involve an impractically large number of rounds.
In this section we first propose a protocol for true shuffling that requires only $O(\sigma)$ communication by the average client, but does so at the expense of a few clients still doing $O(n)$ work. We call this protocol {\em amortized shuffler}. Next, using the amortized shuffler as a subprotocol, we show how the alternating shuffler can be implemented with $O(\sqrt{n})$ communication for each client and with a number of rounds that is at least plausible for production. In the sequel, we refer to the protocol implementing the alternating shuffler functionality as the {\em alternating shuffler protocol}.    

\paragraph{Communication model.}
We assume a server, denoted $\server$ and $n$ clients organized in a start network.
Our protocol starts with a setup round where clients
share a public key with the server,
who subsequently informs clients of the public keys of the clients
with whom they need to communicate privately. From then on the server can act as a relay for private communication between clients. From a security perspective, as we consider a semi-honest 
threat model for the server and non-adaptive client corruptions, this corresponds to the setting where clients can
communicate privately, but {\em the adversary observes the communication pattern}.

\paragraph{Round advancement.}
Our protocol proceeds in rounds, where each round is initiated by the server
with a message to the clients that have not dropped out so far.
The server then receives responses
from clients during the duration of the round (a predefined timeout, or until all alive clients report or explicitly drop out)
and then initiates the next round by sending the next messages.
The server discards all client messages that arrive late, i.e. intended 
for previous round, or malformed. Nevertheless these messages are incorporated to the
server's view when proving security.

\subsection{Components}

We start by introducing three building blocks we will need for both of these protocols. The server $\server$ coordinates the computation among $n$ clients organized in pairwise disjoint committees $C_1\ldots, C_m$. By $c_{i,j}$ be denote the client that is the $j$ member of the $i$th committee. We assume that clients have shared a public key with the server, who subsequently informs clients of the public keys of the clients in their committee and the next committee (clients in committee $C_m$ only receive keys for their committee). Note that ensuring that clients talk to a sublinear number of neighbors is required for our goal of sublinear communication.
Recall that we aim to withstand
a semi-honest server possibly colluding
with up to $\gamma n$ {\em fully malicious} clients. In terms of correctness, our protocols enjoy {\em guaranteed output delivery} to the server, as long as no more than 
a fraction $\alpha$ of the clients
drop out or misbehave.

\subsubsection{Distributed (and replicated) key agreement} 
The first component is a secure protocol for distributing {\em independent} Shamir secret sharings of {\em the same} randomly generated secret $\sk$ across each of a large number of disjoint committees (where $\sk$ is only recoverable if a threshold $t$ number of clients {\em in the same committee} reveal their shares). Our protocol also simultaneously outputs $g^\sk$ in a group of our choice (in which discrete logs are hard) to the server.
This allows any committee to decrypt a subset of ciphertexts, which in combination with provable local shuffling and re-encryption allows to amortize mixing work across clients.
This is achieved while preventing corrupted clients from different committees to collude for decryption,
as the sharings are independent.
A key aspect of our protocol is that clients only incur $O(\committeesize)$ costs (both communication and computation), where $\committeesize$ denotes the committee size. This is a crucial step towards shuffling with sublinear per-client costs. We believe this protocol is novel.

Our protocol is presented in Figure~\ref{Protocol:KeyAgreement}. The basic idea is to have clients in the first committee $C_1$ generate the output key $\sk\in \mathbb{Z}_p$ by exchanging Shamir shares of a locally generated random number. This is the standard approach if we wanted a different key per committee. The naive approach from here would be to have $C_1$ re-share shares of $\sk$
across committees, but this would 
require too much communication for
clients in $C_1$. An alternative would be to have $C_1$ 
re-share $\sk$ with a small number of committees who with in turn re-share with the rest, but this introduces
the need for $\Omega(1)$ communications rounds. In contrast, our solution has $O(\committeesize)$ cost per clients and $4$ rounds. The idea is as follows: Every committee $C_i$ computes a share of a random secret number $s_i\mathbb{Z}_p$ and secret-shares it both across clients in $C_i$ and $C_{i+1}$ (steps \ref{item:share-shares-begin}-\ref{item:share-shares-end}).
For each committee $i>1$, clients then securely reveal to the server the offset $s_i - s_{i-1}$,
who replies back with the offset $d_i := s_i - s_1$ (step~\ref{item:send-offset}). Then clients in $C_i$, for $i>1$, simply update their share of $s_i$ to be shares of $s_i - d_i = s_i - (s_i - s_1) = s_1$, as intended.

The above idea is robust to clients dropping out (up to a predefined threshold $t$) thanks to Shamir sharing.
However, we also need to handle corrupted clients that might distribute incorrect shares. This allows to attain guaranteed output delivery for the server.
For this purpose, our protocol makes use of the
verifiable secret sharing construction
due to Feldman~\cite{FeldmanVSS}, which ensures that a few malicious parties amongst the committees can't prevent its completion by providing invalid shares.
The idea in Feldman's protocol is to, given an appropriate group $\mathbb{G}$ with generator $g$, have the client provide commitments $g^{s}, g^{a_1} \ldots, g^{a_t}$ to coefficients of the random polynomial $P(x) := s + \sum{i=0}^{t} a_ix^i$ used for Shamir sharing a secret $s$ (a uniform random value in $\mathbb{Z}_p$ in our case).
This is done in step~\ref{item:share-shares-begin} of the protocol, for both polynomials $P$ and $Q$.
Then the server can derive a commitment to a given share homomorphically by manipulating commitments (step~\ref{item:send-commitments}) as $g^{P(j)} = g^s\prod_{i = 1}^{t}g^{a_i j^i}$ (for the $j$th share). Recipient clients for the share can then check that the received share and the commitment to it computed by the server match (step~\ref{item:check-commitments}), thus verifying honest sharing of the secret $s$.
If clients find invalid shares they report them to the server in step~\ref{item:check-commitments}, who checks that they're indeed invalid (this important to prevent malicious clients to frame other clients).
This can be easily done by checking the reported share against the ciphertext that the server collected in step~\ref{item:relay-messages}. Note corrupt clients might not report bad shares, but this is equivalent to using bad shares during decryption, which we discuss and address next.

\paragraph{Distributed (verifiable) decryption.} The second component (Figure~\ref{Protocol:Decryption}) is a protocol that allows the parties within one committee to enable the server to decrypt ElGamal ciphertexts encrypted with public key $g^\sk$. This can be done classically requiring each client to receive one group element, do one exponentiation, and send one group element for each ciphertext their committee decrypts. To see how consider an ElGamal ciphertext $c = (mh^\sk, h)$, for message $m$ and group element $h\in\G$, and recall that each cohort members $i$ hold a shamir share $\sk_i$ of $\sk$. The server can just send $h$ to all cohort members,
who then reply with $v_i := h^{\sk_i}$ (see step~\ref{item:partial-dec} in Figure~\ref{Protocol:Decryption}). Since polynomial interpolation is a linear
operation, the server can reconstruct $h^{\sk}$ from the $v_i$'s using group operations and recover $m$ (step~\ref{item:dec}). Note that malicious clients may not construct $v_i$ correctly. To address this, the protocol includes a zero-knowledge proof that the client has behaved correctly, i.e. showing that the reported share match the commitments obtained by the server in the key generation stage (Figure~\ref{Protocol:KeyAgreement}). This is a very efficient variant of a Schnorr proof that requires a single group element per client (and
thus  roughly doubles communication from the client to the server). While it requires an extra round, it can be removed using the Fiat-Shamir heuristic.

\paragraph{Verifiable shuffles.} The third
component (Figure~\ref{Protocol:CTShuffle})
allows the server to obtain a random shuffle of a set of ciphertexts. The intuition behind the protocol is that server sends the set of ciphertexts to clients in turns, to ensure (up to negligible probability) that at least one honest client gets a chance to properly shuffle the ciphertexts.
To ensure correctness in the face of malicious clients, 
our protocol employs a zero-knowledge proof
due to Bayer and Groth~\cite{BayerG12}. Their protocol allows to permute and re-encrypt a collection of ElGamal ciphertexts whilst also providing a zero-knowledge proof that the output is a re-encrypted permutation of the input. For $N$ ciphertexts this protocol can be implemented with communication overhead any fixed constant times $N$ asymptotically and $O(N)$ computations. Alternatively it can provide a communication overhead of only $O(\sqrt{N})$ if one is willing to use $O(N\log(N))$ computation. They report achieving $0.7MB$ of communication with a little over two minutes of computation time to process $10^5$ ciphertexts (including one verification), we expect these numbers to decrease roughly linearly with $N$ down to about $N=1000$. The protocol is compatible with the Fiat-Shamir heuristic which we suggest using to keep interaction to a minimum in out protocol. We also note that some of the required exponentiations can be done before the prover has the data.

\begin{figure}[t]
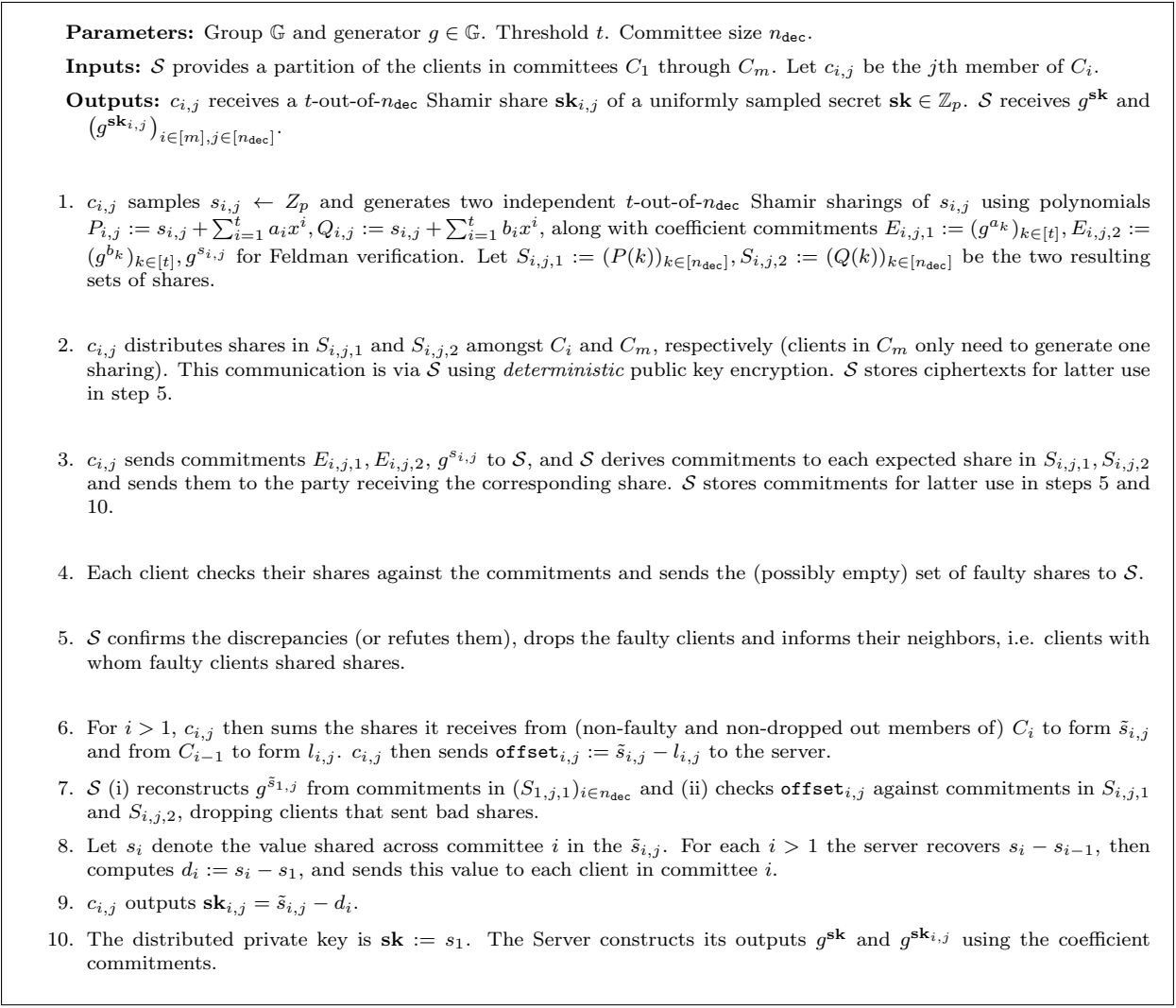

\footnotesize
\begin{framed}
\begin{enumerate}
\item[] \hspace{-0.4cm} \textbf{Parameters:}
Group $\G$ and generator $g\in \G$. Threshold $t$. Committee size $\committeesize$.
\item[] \hspace{-0.4cm} \textbf{Inputs:}
$\server$ provides a partition of the clients in committees $\Committee_1$ through $\Committee_m$. Let $c_{i,j}$ be the $j$th member of $\Committee_i$.
\item[] \hspace{-0.4cm} \textbf{Outputs:}
$c_{i,j}$ receives a $t$-out-of-$\committeesize$ Shamir share $\sk_{i,j}$ of a uniformly sampled secret $\sk\in \mathbb{Z}_p$. $\server$ receives $g^{\sk}$ and $\big(g^{\sk_{i,j}}\big)_{i\in[m],j\in[\committeesize]}$.
\\[2ex]
\item $c_{i,j}$ samples $s_{i,j}\leftarrow {Z}_p$ and generates two independent $t$-out-of-$\committeesize$ Shamir sharings of $s_{i,j}$ using polynomials $P_{i,j} := s_{i,j} + \sum_{i=1}^{t} a_ix^i, Q_{i,j} := s_{i,j}+\sum_{i=1}^{t} b_ix^i$, along with coefficient commitments $E_{i,j,1} := (g^{a_k})_{k\in [t]}, E_{i,j,2} := (g^{b_k})_{k\in [t]}, g^{s_{i,j}}$ for Feldman verification.
Let $S_{i,j,1} := (P(k))_{k\in [\committeesize]},S_{i,j,2}:=(Q(k))_{k\in [\committeesize]}$ be the two resulting sets of shares.\\[1ex]\label{item:share-shares-begin}
\item $c_{i,j}$ distributes shares in $S_{i,j,1}$ and $S_{i,j,2}$ amongst $\Committee_i$ and $\Committee_m$, respectively (clients in $\Committee_m$ only need to generate one sharing). This communication is via $\server$ using {\em deterministic} public key encryption. $\server$ stores ciphertexts for latter use in step~\ref{step:resolve-discrepancies}.\\[1ex]\label{item:relay-messages}
\item $c_{i,j}$ sends commitments $E_{i,j,1} ,E_{i,j,2}$, $g^{s_{i,j}}$ to $\server$, and $\server$ derives commitments to each expected share in $S_{i,j,1},S_{i,j,2}$ and sends them to the party receiving the corresponding share. $\server$ stores commitments for latter use in steps~\ref{step:resolve-discrepancies} and ~\ref{step:output}.\\[1ex]\label{item:send-commitments}
\item Each client checks their shares against the commitments and sends the (possibly empty) set of faulty shares to $\server$.\\[1ex]\label{item:check-commitments}
\item $\server$ confirms the discrepancies (or refutes them), drops the faulty clients and informs their neighbors, i.e. clients with whom faulty clients shared shares.\label{step:resolve-discrepancies}\\[1ex]
\label{item:share-shares-end}
\item For $i>1$, $c_{i,j}$ then sums the shares it receives from (non-faulty and non-dropped out members of) $\Committee_i$ to form $\tilde{s}_{i,j}$ and from $\Committee_{i-1}$ to form $l_{i,j}$. $c_{i,j}$ then sends $\texttt{offset}_{i,j} := \tilde{s}_{i,j}-l_{i,j}$ to the server. 
\item $\server$ (i) reconstructs $g^{\tilde{s}_{1,j}}$ from commitments in $(S_{1,j,1})_{i\in\committeesize}$ and (ii) checks $\texttt{offset}_{i,j}$
against commitments in $S_{i,j,1}$ and $S_{i,j,2}$, dropping clients that sent bad shares.
\item Let $s_i$ denote the value shared across committee $i$ in the $\tilde{s}_{i,j}$. For each $i>1$ the server recovers $s_i-s_{i-1}$, then computes $d_i := s_i-s_1$, and sends this value to each client in committee $i$.\label{item:send-offset}
\item $c_{i,j}$ outputs $\sk_{i,j}=\tilde{s}_{i,j}-d_i$.
\item The distributed private key is $\sk:=s_1$. The Server constructs its outputs $g^{\sk}$ and $g^{\sk_{i,j}}$ using the coefficient commitments.\label{step:output}
\end{enumerate}
\end{framed}
\caption{\small Shared Key Agreement}
\label{Protocol:KeyAgreement}
\end{figure}

\begin{figure}[t]
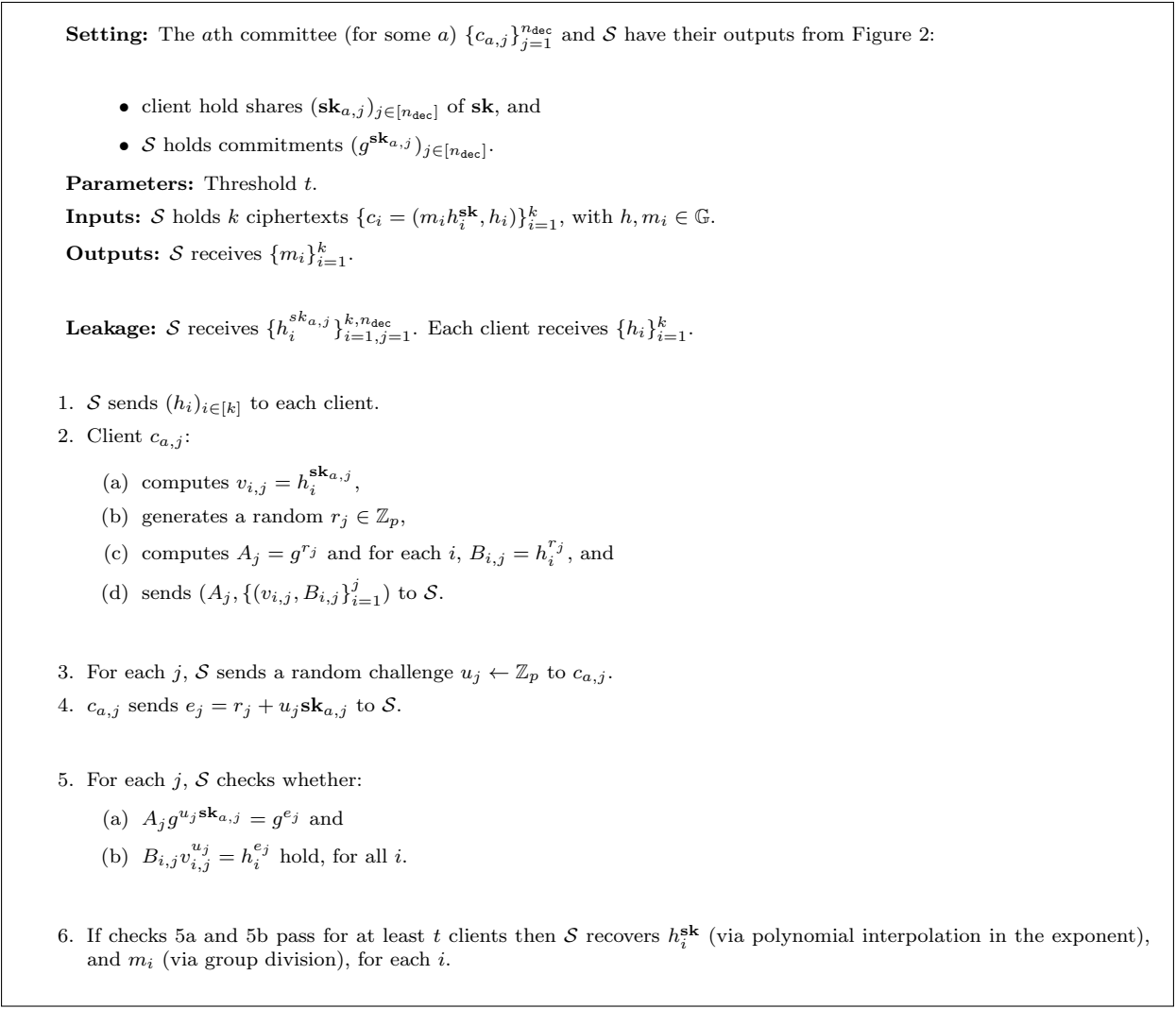

\footnotesize
\begin{framed}
\begin{enumerate}
\item[] \hspace{-0.4cm} \textbf{Setting:}
The $a$th committee (for some $a$) $\{c_{a,j}\}_{j=1}^{\committeesize}$ and $\server$ have their outputs from Figure~\ref{Protocol:KeyAgreement}:\\[1ex]
\begin{itemize}
    \item client hold shares $(\sk_{a,j})_{j\in[\committeesize]}$ of $\sk$, and 
    \item $\server$ holds commitments $(g^{\sk_{a,j}})_{j\in[\committeesize]}$.
\end{itemize}
\item[] \hspace{-0.4cm} \textbf{Parameters:}
Threshold $t$. \item[] \hspace{-0.4cm} \textbf{Inputs:}
$\server$ holds $k$ ciphertexts $\{c_i = (m_i h_i^{\sk},h_i)\}_{i=1}^k$, with $h,m_i\in \G$.
\item[] \hspace{-0.4cm} \textbf{Outputs:}
$\server$ receives $\{m_i\}_{i=1}^k$.\\[2ex]
\item[] \hspace{-0.4cm} \textbf{Leakage:}
$\server$ receives $\{h_i^{sk_{a,j}}\}_{i=1,j=1}^{k,\committeesize}$. Each client receives $\{h_i\}_{i=1}^k$.
\\[2ex]
\item $\server$ sends $(h_i)_{i\in [k]}$ to each client.
\item Client $c_{a,j}$:
\begin{enumerate}
    \item computes $v_{i,j} = h_i^{\sk_{a,j}}$, \label{item:partial-dec}
    \item generates a random $r_j\in \mathbb{Z}_p$,
    \item computes $A_j = g^{r_j}$ and for each $i$, $B_{i,j} = h_i^{r_j}$, and
    \item sends $(A_j,\{(v_{i,j},B_{i,j}\}_{i=1}^j)$ to $\server$.\\[2ex]
\end{enumerate}
\item For each $j$, $\server$ sends a random challenge $u_j\leftarrow \mathbb{Z}_p$ to $c_{a,j}$.
\item $c_{a,j}$ sends $e_j = r_j+u_j\sk_{a,j}$ to $\server$.\\[2ex]
\item For each $j$, $\server$ checks whether:
\begin{enumerate}
    \item $A_jg^{u_j\sk_{a,j}}=g^{e_j}$ and \label{item:check-1}
    \item $B_{i,j}v_{i,j}^{u_j} = h_i^{e_j}$ hold, for all $i$.\label{item:check-2}\\[2ex]
\end{enumerate}
\item If checks~\ref{item:check-1} and~\ref{item:check-2} pass for at least $t$ clients then $\server$ recovers $h_i^{\sk}$ (via polynomial interpolation in the exponent),
and $m_i$ (via group division), for each $i$. \label{item:dec}
\end{enumerate}
\end{framed}
\caption{\small Shared Key Decryption}
\label{Protocol:Decryption}
\end{figure}

The following theorem states our security guarantees for 
the key generation and decryption protocols. The result is stated in terms of a per-committee guarantee, and parameterized by a threshold $t$ which will be chosen to optimize the resulting theorems. As discussed above, we have security if no more than a given fraction of the clients is malicious, and additionally correctness in executions where sufficiently many honest clients follow the protocol.
Our proofs
are in the standard simulation-based security model~\cite{howtosimulate} of multi-party computation. This theorem is proved in Appendix~\ref{sec:appendix-implementations}.

\begin{theorem}\label{thm:sharedkeyprotocols}
  The protocols of Figures~\ref{Protocol:KeyAgreement}~and~\ref{Protocol:Decryption} securely compute (with abort) the functionalities described by their inputs and outputs (and in the decryption case leakage) against an adversary consisting of a semi-honest server and up to $t-1$ malicious clients in each committee. They also guarantee output if at least $t$ clients in each committee follow the protocol.
\end{theorem}

\begin{figure}[t]
\footnotesize
\begin{framed}
\begin{enumerate}
\item[] \hspace{-0.4cm} \textbf{Setting:}
A collection of $\shufflecommitteesize$ clients. Typically $\shufflecommitteesize = O(\sigma)$ for a statistical security parameter $\sigma$.
\item[] \hspace{-0.4cm} \textbf{Parameters:}
Dropout limit $d$.
\item[] \hspace{-0.4cm} \textbf{Inputs:}
The server $\server$ holds ElGamal ciphertexts
$S:=\{m_i h_i^{\sk},h_i\}_{i\in[k]}$ 
and the corresponding public key $\pk=g^\sk$.
\item[] \hspace{-0.4cm} \textbf{Outputs:}
$\server$ receives a shuffled re-encryption of ciphertexts in $S$.
\item Let $L$ be the collection of clients
in a random order chosen by $\server$.
\item For $\shuffleclient\in L$ in the collection of clients:
\begin{enumerate}
    \item $\server$ sends the current ciphertexts set $S$ to $\shuffleclient$.
    \item Client $\shuffleclient$ \label{item:client-shuffles}
    \begin{itemize}
        \item sets $S_{\shuffleclient} := \pi_{\shuffleclient}(\texttt{re-encrypt}(S))$, where $\pi$ to $S$ is a uniformly random permutation, and
        \item constructs a proof $P$ that it has done so (c.f. Bayer \& Groth~\cite{BayerG12}).
        \item $\shuffleclient$ sends $(S_{\shuffleclient}, P)$ to $\server$.
    \end{itemize}
    \item If $\server$ receives $(S_{\shuffleclient},P)$ successfully, and $P$ is valid, then it sets $S:=S_{\shuffleclient}$. Otherwise it leaves $S$ unchanged and drops $\shuffleclient$. %
    \item If $\shufflecommitteesize-d$ clients have provided valid shuffles go to step~\ref{item:output-S}.
    \item If $d+1$ clients have failed to provide valid shuffles abort the protocol.
\end{enumerate}
\item $\server$ outputs the current set of ciphertexts $S$.\label{item:output-S}
\end{enumerate}
\end{framed}
\caption{\small Ciphertext Shuffle}
\label{Protocol:CTShuffle}
\end{figure}

\subsection{Amortized Shuffler}

With these three building blocks in place there is a simple protocol to provide the server with shuffled inputs. First we perform the distributed key agreement, then each client encrypts their input with the resulting key, then a few of the clients are selected to each shuffle the ciphertexts and prove they have done so correctly, then the server gets the clients to decrypt the resulting values.
The shuffler also adds an offset to the public key that is used. This is done so as to protect against another possible adversary who controls more than $\gamma$ fraction of clients so long as the server is honest.
We describe this protocol in Figure~\ref{Protocol:Amortized},
which we call {\em amortized shuffler}: the reason is that the protocol amortizes shuffling costs across clients, 
resulting in sublinear, i.e. $O(\log n)$,
costs for the average clients. First, we discuss the security and correctness properties of the protocol, and discuss costs in detail next. The following theorem is proved in Appendix~\ref{sec:appendix-implementations}.

\begin{figure}[t]
\footnotesize
\begin{framed}
\begin{enumerate}
\item[] \hspace{-0.4cm} \textbf{Parameters:}
Number of shufflers $\shufflecommitteesize$, dropout limit $d$.
\item[] \hspace{-0.4cm} \textbf{Inputs:}
Each client $c$ has an input $m_c$.
\\[1ex]
\begin{center}
    \underline{Distributed Key Generation}
\end{center}
\item $\server$ partitions clients into $m$ committees of size $\committeesize$ uniformly at random. All parties then perform the key generation in Figure~\ref{Protocol:KeyAgreement}. $\server$ generates $\sk'$ at random, sets $\sk''=\sk + \sk'$ and sends the resulting key $g^{\sk''}=g^{\sk}g^{\sk'}$ to all clients to use as a public key.
\item Each client $c$ computes the ciphertext  $(m_cg^{\sk'' r},g^r)$ and sends it to the server.
\\[1ex]
\begin{center}
    \underline{Distributed Verifiable Shuffling}
\end{center}
\item $\server$ runs the protocol of Figure~\ref{Protocol:CTShuffle} with $\shufflecommitteesize$ randomly selected clients and dropout limit $d$, all the ciphertexts and the new public key $\pk'=g^{\sk+t}$.\label{item:shuffle-amortized}
\\[1ex]
\begin{center}
    \underline{Distributed Decryption}
\end{center}
\item The server uses key homomorphism to return the key to $\sk$ from $\sk''$.
\item $\server$ partitions ciphertexts into $m$ groups $\{G_i\}_{i=1}^m$ and for each $i$ runs Figure~\ref{Protocol:Decryption} with $\Committee_i$ and $G_i$.
\item $\server$ takes all the resulting plaintexts as output.

\end{enumerate}
\end{framed}
\caption{\small Amortized Shuffler}
\label{Protocol:Amortized}
\end{figure}

\begin{theorem}\label{thm:amortizedsecurity}
The Amortized shuffler (Figure~\ref{Protocol:Amortized}) securely implements the shuffling functionality against an adversary consisting of a semi-honest $\server$, up to $\gamma n$ malicious clients, and the ability to drop honest clients actively, with statistical security parameter $\sigma$ given by the smaller of 
\begin{equation*}
    -\log_2(m) + 2\log_2(e)(t/\committeesize-\gamma)^2\committeesize - 1
\end{equation*}
and
\begin{equation*}
    2\log_2(e)(1-d/\shufflecommitteesize-\gamma)^2\shufflecommitteesize) - 1.
\end{equation*}

Furthermore, so long as the server is semi-honest and at least 
$(1-\alpha)n$ non-actively selected clients follow the protocol without dropping out, the probability of aborting is at most $2^{-\eta}$ where $\eta$ is the smaller of
\begin{equation*}
    -\log_2(m) + 2\log_2(e)((t+1)/\committeesize-(1-\alpha))^2\committeesize - 1
\end{equation*}
and
\begin{equation*}
    2\log_2(e)((d+1)/\shufflecommitteesize-\alpha)^2\shufflecommitteesize) - 1.
\end{equation*}
The protocol is also secure against an adversary controlling an arbitrary number of clients, so long as the server is honest. Using the Fiat-Shamir heuristic to create all uniform challenges the protocol runs in $\shufflecommitteesize + 5$ rounds.
\end{theorem}

If we require $\sigma=40$ and $\eta=10$ with $n=10,000$ clients at most 500 of which drop out and at most 500 of which are malicious then we can achieve a 24 rounds (or 18 if no-one drops out while shuffling) a worst case communication of $1.25MB$ and an average communication of $4.35KB$. These figures correspond to an optimized numerical analysis slightly tighter than the above formulae. In our feasibility study we estimate that each shuffling client will only be required to do 2 or 3 seconds of computation, thus with fast round trip times the whole protocol could take less than a minute.  The asymptotics implied by the above are given in the following theorem.

\begin{theorem}\label{thm:amortized}
If the server is semi-honest, the fraction of malicious clients $\gamma$ and dropouts bounded by $\alpha$ with $2\gamma + \alpha < 1$, then the amortized shuffler does, with an appropriate choice of parameters, the following. For simplicity of expressions we assume $\sigma < n$.
\begin{itemize}
    \item Implement shuffling with statistical security $\sigma$ and correctness parameter $\eta$.
    \item Require at most $O(n)$ communication from any client and $O(\sigma/(1-\gamma-\alpha)^2)$ communication from the average client.
    \item Require at most $O(n)$ exponentiations from any client and $O(\sigma/(1-\gamma-\alpha)^2+\log(n))$ exponentiations from the average client.
    \item Requires $O(\sigma/(1-\gamma-\alpha)^2)$ rounds.
\end{itemize}
\end{theorem}

\subsection{Alternating Shuffler}

\begin{figure}[t]
\footnotesize
\begin{framed}
\begin{enumerate}
\item[] \hspace{-0.4cm} \textbf{Parameters:}
Number of rounds $l$, client arrangement $\pi$, height and width $h,w$ such that the number of clients $n=hw$.
\item[] \hspace{-0.4cm} \textbf{Inputs:}
Each client $c$ has an input $m_c$.
\\[1ex]
\begin{center}
    \underline{Distributed Key Generation}
\end{center}
\item The server assigns the clients into $m$ decryption committees of size $\committeesize$ uniformly at random. All parties then perform the key generation in Figure~\ref{Protocol:KeyAgreement} and the server sends the resulting key $g^\sk$ to all clients.
\item Each client $c$ computes $\Enc_{\sk}(m_c)$ and sends it to the server.
\item The server assigns the ciphertexts into an $h\times w$ array $\mat{M}$ according to the permutation $\pi$.
\\[1ex]
\begin{center}
    \underline{Distributed Verifiable Shuffling}
\end{center}
\item $\server$ generates $t$ at random and changes all ciphertexts to be encrypted with $\sk'=\sk+t$, using key homomorphism, and distributes the new public key $g^{\sk+t}$.\\[1ex]
\item $\server$ splits the clients into shuffling committees of size $\shufflecommitteesize$.\\[1ex]
\item The following is repeated $l$ times:
\begin{enumerate}
    \item $\server$ conducts $h$ parallel instances of the protocol of Figure~\ref{Protocol:CTShuffle} (one per row in $\mat{M}$). Each instance is run with one of the shuffling committees (chosen to spread the load as evenly as possible over those committees), and inputs
     $\sk'$ and ciphertexts in a row in $(\mat{M}_i)_{i\in [h]}$.\\[1ex]
    \item $\server$ replaces each row of ciphertexts in $\mat{M}$ with its shuffled version.
    \item $\server$ transposes $\mat{M}$.\\[1ex]
\end{enumerate}
\begin{center}
    \underline{Distributed Decryption}
\end{center}
\item The server uses key homomorphism to return the key to $\sk$ from $\sk'$.
\item $\server$ splits the ciphertexts in $\mat{M}$ into $m$ groups $\{G_i\}_{i=1}^m$ and for each $i$ runs Figure~\ref{Protocol:Decryption} with $\Committee_i$ and $G_i$.
\item $\server$ takes all the resulting plaintexts as output.

\end{enumerate}
\end{framed}
\caption{\small Alternating Shuffler Protocol}
\label{Protocol:Alternating}
\end{figure}

If we use the ciphertext shuffle to shuffle subsets of the values we get the protocol in Figure~\ref{fig:alternatin-shuffler} which implements the Alternating shuffle. The security considerations for this are summarized in the following theorem (proved in Appendix~\ref{sec:appendix-implementations}).

\begin{theorem}\label{thm:Alternating}
The Alternating shuffler protocol (Figure~\ref{Protocol:Alternating}) securely implements the shuffling functionality against an adversary consisting of a semi-honest $\server$, up to $\gamma n$ malicious clients and the ability to drop honest clients actively, with statistical security parameter $\sigma$ given by the smaller of 
\begin{equation*}
    -\log_2(m)+2\log_2(e)(t/\committeesize-\gamma)^2\committeesize-1
\end{equation*}
and
\begin{equation*}
    -\log_2(h\lceil l/2 \rceil + w \lfloor l/2 \rfloor)+2\log_2(e)(1-d/\shufflecommitteesize-\gamma)^2\shufflecommitteesize)-1.
\end{equation*}

Furthermore, so long as the server is semi-honest and at least $(1-\alpha)n$ non-actively selected clients follow the protocol without dropping out, the probability of aborting is at most $2^{-\eta}$ where $\eta$ is the smaller of
\begin{equation*}
    -\log_2(m)+2\log_2(e)((t+1)/\committeesize-(1-\alpha))^2\committeesize - 1
\end{equation*}
and
\begin{equation*}
    -\log_2(h\lceil l/2 \rceil + w \lfloor l/2 \rfloor)+2\log_2(e)((d+1)/\shufflecommitteesize-\alpha)^2\shufflecommitteesize) - 1.
\end{equation*}
\end{theorem}

If we require $\sigma=40$ and $\eta=10$ with $n=10,000$ clients sending 128-bit inputs then at most 500 of which drop out and at most 500 of which are malicious then we can achieve a 47 rounds (or 35 if no-one drops out while shuffling) a worst case communication of $26KB$ and an average communication of $7.5KB$. As in the amortized case these are numerically optimized. Based on a count of exponentiations the computation for each shuffling client should take around 0.03 seconds and we note that the $\sqrt{n}$ shuffles happening in parallel needn't progress through rounds in lockstep with each other. Thus if the round trip times are negligible and all clients respond as quickly as one might hope (admittedly quite optimistic assumptions) the whole protocol would take less than two seconds, from this we conclude that client computation is not a problem. The asymptotics implied by the above are given in the following theorem.

\begin{theorem}\label{thm:alternatingasymptotics}
If the server is semi-honest, the fraction of malicious clients $\gamma$ and dropouts bounded by $\alpha$ with $2\gamma + \alpha < 1$, then the alternating protocol, with an appropriate choice of parameters, does the following. For simplicity of expressions we assume $\sigma/(1-\gamma-\alpha)^2 < \sqrt{n}$.
\begin{itemize}
    \item Implements alternating shuffling with statistical security $\sigma$ and correctness parameter $\eta$.
    \item Requires at most $O(\sqrt{n})$ communication from each client.
    \item Requires at most $O(\sqrt{n})$ exponentiations from any client.
    \item Requires $O(l \sigma/(1-\gamma-\alpha)^2)$ rounds.
\end{itemize}
\end{theorem}

Taking $l=2$ in the above gives the asymptotics for shuffling each row twice. Taking $l=\sigma$ gives the asymptotics for using this protocol to implement a true shuffle using Theorem 3.6 in Håstad~\cite{Hastad}.

\subsection{Feasibility Study}\label{sec:feasibility}
In this section we discuss the concrete efficiency of our protocol, and its amenability for use in production. We first discuss communication costs, including round complexity, and then turn our attention to client's computation requirements. Then we discuss the end-to-end throughput of our protocols. Throughout the section we assume an elliptic curve based implementation, and in particular use curve Curve25519 in our runtime benchmarks. Therefore, a group element can be represented in 256 bits.

\paragraph{Communication costs.}
Figure~\ref{fig:plots-parameters}(left 2 plots) shows how the number of rounds of our protocols grow with the security parameter $\sigma$, for $\sigma=13$, matching the DP guarantee of Gordon et al., and $\sigma=40$, which is considered a negligible probability in practice when dealing with statistical security. Note that the number of rounds of our protocol varies with the client dropout rate in a given execution. In the plot we show the two extremes: the optimistic number of rounds (when no dropouts happen), and the pessimistic number, where the number of dropouts is as large as the protocol tolerates ($n/20$ in the plot, where $n$ stands for the number of clients) and these drop outs happen in the worst possible manner: a client that is selected for shuffling drops out as it obtains the encrypted data (step~\ref{item:shuffle-amortized} of Figure~\ref{Protocol:Amortized}, and step~\ref{item:client-shuffles} of Figure~\ref{Protocol:CTShuffle}).

Recall that while our protocols might take a few rounds of interaction between clients and the server, not all clients need to interact in every round. In fact only one client does work in most rounds of the Amortized protocol (again step~\ref{item:shuffle-amortized} of Figure~\ref{Protocol:Amortized}, and step~\ref{item:client-shuffles} of Figure~\ref{Protocol:CTShuffle}). Figure~\ref{fig:plots-parameters}(right 2 plots) show
total (including both download and upload) per-client communication of our protocols. Note that for the amortized shuffler, shuffling $10^6$ and $10^4$ values requires worst-case communication of $100MB$ and $1MB$, respectively. For the $2$-round alternating shuffler protocol, sublinear costs really make a difference: $n=10^6$ and $n=10^4$ the worst-case client communication is $15KB$ and $130KB$, respectively. The average client incurs significantly less communication in both the amortized and alternating shuffler. The strange drops in the worst-case communication for the Alternating Shuffler protocol  due to the number of shuffles that need doing going from a little more than the number of clients to a little less, meaning the worst of client suddenly goes from having to do two shuffles to doing only one.

\begin{figure*}
    \centering
    \includegraphics[width=4cm]{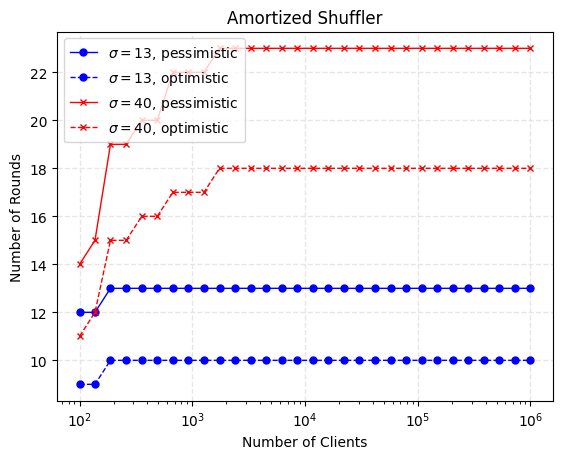}
    \includegraphics[width=4cm]{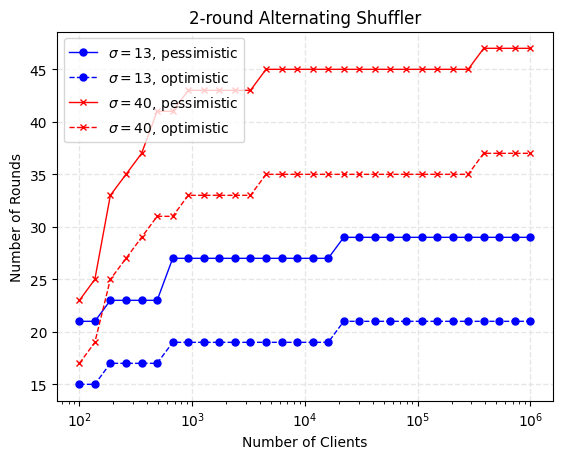}
    \includegraphics[width=4cm]{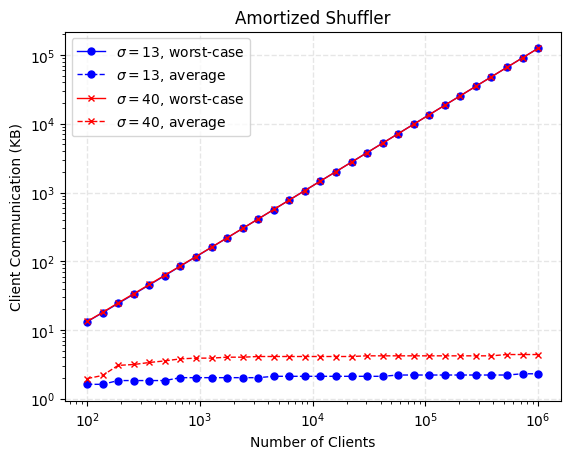}
    \includegraphics[width=4cm]{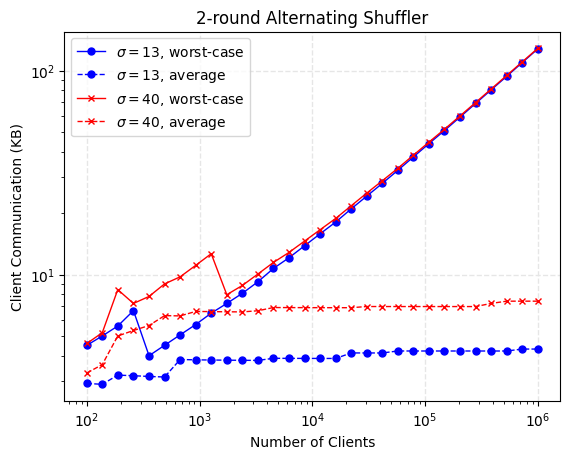}
    \caption{\small (left) Number of rounds as a function of number of clients, for maximum dropout fraction $\alpha = 1/20$, maximum client corruption fraction $\gamma = 1/20$, and correctness parameter $\eta = 10$. (right) Total (download+upload) per-client communication cost (in KB) as n grows, again for $\alpha = 1/20, \gamma = 1/20, \eta = 10$.}
    \label{fig:plots-parameters}
\end{figure*}

\paragraph{Computation costs.}
Next, we characterize client costs in terms of the number of exponentiations that clients need to perform, as these are the most costly operations in our protocol. Table~\ref{tab:exponentiations} shows that count. Clearly the shuffle proof dominates the cost. For reference, 
we run benchmarks for exponentiation 
in the (standard) elliptic curve Curve25519,
using the Dalek-Cryptography framework~\cite{daleklib}, written in Rust. We benchmarked on both a standard laptop, and a Pixel7 device.
A single exponentiation in a Pixel7 phone and a standard laptop take $0.045$ms and $0.0115$ms, respectively.
This means even the worst-case costs in 
Table~\ref{tab:exponentiations} remain under $30$s, even when the client runs in a phone.
Moreover, this estimate is conservative, as it does not account for the fact that the group exponentiations in the Bayer-Groth shuffling proof can be structured as multi-exponentiations, resulting in a significant speed-up via efficient algorithms such as Pippenger's~\cite{Pippenger80}. For instance, in our benchmarking we get $3.4$ and $5.2x$ speedups on a Pixel 7 device for $n = 10^3$ and $n = 10^4$, respectively. In Table~\ref{tab:exponentiations-ms} (Appendix~\ref{sec:appendix-runtimes})
we report the results of our benchmark, both for Pixel 7 and laptop. Concretely, Table~\ref{tab:exponentiations-ms} is analogous to  Table~\ref{tab:exponentiations}, but reporting time in milliseconds.
For amortized shuffling we report timing using Pippenger's algorithm.

\begin{table}
    \centering
    \begin{tabular}{c|cc|cc}
        Number of Clients & \multicolumn{2}{c}{$10^3$} & \multicolumn{2}{c}{$10^5$}\\ \hline
        Key Agreement & \multicolumn{2}{c}{28} & \multicolumn{2}{c}{32} \\
        Decryption & \multicolumn{2}{c}{39} & \multicolumn{2}{c}{43} \\
        & \textsc{Avg} & \textsc{Worst} & \textsc{Avg} & \textsc{Worst} \\\cline{2-5}
        Shuffling (Amortized) & 105 & 6198 & 108 & 600372 \\
        Shuffling (Alternating) & 331 & 544 & 265 & 2086
    \end{tabular}
    \vspace{1cm}
    \caption{\small Exponentiations per client in each part of the protocol . All values are for $\sigma=40$, $\gamma =1/20$, $\alpha=1/20$, $\eta= 10$. Shuffling rows display counts averaged over all clients and for the client with the worst load.}
    \label{tab:exponentiations}
\end{table}

\paragraph{Throughput.}
Finally, we discuss end-to-end-costs of our protocols,
with a focus on client costs.
The more expensive part is the shuffling stage,
where selected clients have to (i) receive
ciphertexts from the server, (ii) shuffle, reencrypt them, and
compute a Bayer-Groth proof of shuffling, and (iii) send 
the ciphertexts and proofs back to the server.
Consider the amortized shuffler (costs are significantly better for $2$ round alternating), and the case of $n = 10^4, \alpha = 1/20, \gamma = 1/20, \eta = 10, \sigma = 40$. As shown in Figure~\ref{fig:plots-parameters}, the total worst-case (download+upload) communication is 1MB, and therefore network cost would stay below 1s, even with a slow (conventional) connection.
As a conservative estimate, computation would take less than 5s, according to Table~\ref{tab:exponentiations-ms} (as we're considering $n=10^4$). Since for $\sigma=40$ the amortized shuffler requires 
$\sim$20 rounds, the total client's time for $n = 10^4$ would be no more than 2 minutes. 

A final remark regarding the different roles clients might take in our protocols is in order. The way our protocols are presented, 
clients always encrypt and decrypt, and possibly act as shufflers.
If we could assume that shufflers are always powerful devices, e.g. desktops, then the end-to-end costs would significantly improve (our benchmarks show a 4x improvement from Pixel 7 to laptop).
Another remarkable aspect of the protocol is that, while the average 
client does not do a lot of work, it has to wait until the end of the protocol for decryption. A conceivable protocol variant would involve clients that do not act as decryptors, and instead only receive 
a public key after the key agreement phase, encrypt their message, and terminate the execution, thus incurring only one round of interaction. In that case, the thread model assumptions regarding
malicious clients and dropouts would be with respect to the 
clients taking on encryption/decryption and shuffling tasks.

\section{Conclusion}
We have shown that by considering simultaneously the concrete means of amplification, e.g.\ uniform shuffling vs approximate shuffling, and their cryptographic implementation in a realistic threat model, we can strike a balance between privacy and computational costs, shedding some light on practical aspects of the "amplification by shuffling" paradigm.

Our results so far leave open some intriguing problems.
Although we do not solve them in this work, we 
discuss them briefly next.

\paragraph{Does Public Randomness and Two Rounds give Strong Amplification?}
In Section~\ref{sec:no-strong-amp} we showed that 
2 rounds of alternate shuffling can't be used to achieve 
strong amplification, that is, it is insufficient to achieve the 
amplification bound of Theorem~\ref{thm:uniform-shuffling}. However, the proof relies on a specific configuration that is, intuitively, unlikely
to happen given the initial
randomization applied by $\pi$, even if $\pi$
is public. Leveraging this fact might be enough to achieve strong amplification without changing our cheapest shuffling protocol.

\paragraph{Is Three Round Alternating DO with Good $(\epsilon, \delta)$?}
Another way to achieve strong amplification with
imperfect shuffling
is DO-shuffling, thanks to the result in~\cite{DBLP:journals/iacr/ZhouSCM22} (DO-shuffling enables strong amplification).
As shown in Section~\ref{sec:not-do}
this can't be achieved using a $2$-round alternating  shuffler. Therefore, whether 
3 rounds of alternating shuffler gives DO shuffling with good $(\epsilon, \delta)$ (it's not hard to show that can be achieved $(0, n^{-1/4})$ and $(n\log n, 0)$)
parameters is an interesting next question to tackle.

\paragraph{Can One Obtain Numerically Tight DP Bounds for Alternating Shuffling Protocols?}
The privacy bounds proved in Section~\ref{sec:weak} are illuminating with regards to the amplification power provided by approximate shuffling in comparison to exact shuffling, but the constants obtained are far from tight. Practical deployments of DP often rely on numerical accountants instead of closed-form expressions for computing tight privacy guarantees (e.g.\ \cite{koskela2023numerical}), in particular if the mechanisms need to access the same data multiple times through composition. Designing numerical accountants for approximate shuffling protocols is an important question for future work.

\ifnum\arxiv=0
    \bibliographystyle{ACM-Reference-Format}
\else
    \section*{Acknowledgements}
    The authors want to thank Kobbi Nissim for stimulating conversations at the early stages of this project.
    \bibliographystyle{abbrvnat}
\fi
\bibliography{main}

\newpage
\appendix
\onecolumn

\section{Proofs from Section~\ref{sec:properties}}

\subsection{Proof of Theorem~\ref{fig:alternatin-shuffler}}\label{app:amplification-calculation}

Let us write the number of clients as $n = h \cdot w$ where $h$ and $w$ correspond, respectively, to the number of rows and columns in the alternating shuffler $\alttwo$. At the end we will set $h = w = \sqrt{n}$, but it is illustrative and convenient to keep a separate notation for the numbers of rows and columns throughout the proof.

Consider two neighboring databases $D, D'$ differing (w.l.o.g.) in the first individual.
We start by introducing a coupling between the executions of the protocol on both databases.
First of all, the coupling will select the same permutations $\pi_2, \ldots, \pi_h$ for all rows except the first one (which contains the differing user) in both executions.
The permutation $\pi_1$ applied in the first row is obtained in two steps: 1) shuffle the users $\{2, \ldots, w\}$ among themselves with a random permutation, and 2) swap the first user with a position chosen u.a.r.\ from $\{1,\ldots, w\}$.
We also couple the first step in this decomposition to use the same permutation on both executions.
Note that the position in the first row where the first user is sent is still chosen independently in both executions, and marginally on each execution the permutation of the first row is still uniform.

Let $C_1, \ldots, C_k \in Y^w$ (resp.\ $C_1', \ldots, C_w'$) be the outputs of the protocol on input $D$ (resp.\ $D'$) organized by column (i.e.\ $C_i$ contains the messages in the $i$th column after applying $\alttwo$).
Let $I, I' \in [w]$ denote the random variables indicating the columns where the first user is sent to by the permutation of the first row in both executions. These random variables are independent and uniform.

Now let $C_{<i}$ represent the output of the first $i-1$ columns.
Our proof strategy relies on analyzing the privacy loss incurred by sequentially revealing each column conditioned on the output of the previous columns, $C_i | C_{<i} = c_{<i}$, and then applying a standard (adaptive) composition argument.
To analyze the privacy of each columns we combine privacy amplification by sampling and by shuffling.

First of all, we observe that for any $i$ and any fixed output $c_{<i}$ we have
\begin{align}
    \Pr[I = i | C_{<i} = c_{<i}] \leq \frac{e^{2\epsilon_0}}{e^{2\epsilon_0} + w - 1} := \gamma \enspace.
    \label{eqn:posterior}
\end{align}
Indeed, since a priori $\Pr[I=i] = 1/w$, by Bayes' rule we have
\begin{align*}
    \Pr[I = i | C_{<i} = c_{<i}]
    &=
    \frac{\Pr[C_{<i} = c_{<i}, I=i]}{\Pr[C_{<i} = c_{<i}]}
    \\
    &=
    \frac{\Pr[C_{<i} = c_{<i}, I=i]}{\sum_j \Pr[C_{<i} = c_{<i}, I=j]}
    \\
    &=
    \frac{1}{1 + \sum_{j \neq i} \frac{\Pr[C_{<i} = c_{<i}, I=j]}{\Pr[C_{<i} = c_{<i}, I=i]}}
    \\
    &=
    \frac{1}{1 + \sum_{j \neq i} \frac{\Pr[C_{<i} = c_{<i} | I=j]}{\Pr[C_{<i} = c_{<i} | I=i]}}
    \enspace.
\end{align*}
Now let $A_{|i}$ denote the random assignments $A = (A_1, \ldots, A_n) \in [w]^n$ indicating the column where each user's message is sent to by the first step of the alternating shuffler conditioned on the first user sending their message to the $i$th column. Note that we can couple $A_{|i}$ and $A_{|j}$ in such a way that they differ by exactly two assignments: first sample $A_{|i} = (A_1, \ldots, A_n)$ with $A_1 = i$, and then obtain $A_{|j} = (A_1', \ldots, A_n')$ by taking $A_1' = j$ and letting $A_u' = i$ for $u$ sampled uniformly at random from $\{ u : A_u = j \}$.
Using this coupling we can bound $p := \Pr[C_{<i} = c_{<i} | I=i]$ as:
\begin{align*}
    p
    &=
    \sum_{a} \Pr[C_{<i} = c_{<i} | I=i, A=a] \Pr[A=a | I=i]
    \\
    &=
    \sum_{a} \Pr[C_{<i} = c_{<i} | A=a] \Pr[A_{|i}=a]
    \\
    &=
    \sum_{a, b} \Pr[C_{<i} = c_{<i} | A=a] \Pr[A_{|i}=a, A_{|j}=b]
    \\
    &\leq
    e^{2 \epsilon_0} \sum_{a, b} \Pr[C_{<i} = c_{<i} | A=a'] \Pr[A_{|i}=a, A_{|j}=b]
    \\
    &=
    e^{2 \epsilon_0} \sum_{b} \Pr[C_{<i} = c_{<i} | A=a'] \Pr[A_{|j}=b]
    \\
    &=
    e^{2 \epsilon_0} \Pr[C_{<i} = c_{<i} | I=j] \enspace,
\end{align*}
where we used that $C_{<i} | A=a$ and $C_{<i} | A=b$ represent the output of $w (i-1)$ copies of the $\epsilon_0$-DP local randomizer $R$ applied to a database differing in at most two positions. Thus we obtain \eqref{eqn:posterior}.
The same argument shows the same holds for the other execution: $\Pr[I' = i | C'_{<i} = c_{<i}] \leq \gamma$.

Now we consider, for fixed $c_{<i}$, four distributions associated with the outputs of column $i$ on both executions conditioned on observing the same output from the previous columns.
We split the cases where the first user is sent to column $i$ or not:
\begin{align*}
    {M_i}^{\mathrm{in}} &:= C_i | C_{<i} = c_{<i}, I = i \enspace, \\
    {M_i}^{\mathrm{out}} &:= C_i | C_{<i} = c_{<i}, I \neq i \enspace, \\
    {M'_i}^{\mathrm{in}} &:= C'_i | C'_{<i} = c_{<i}, I' = i \enspace, \\
    {M'_i}^{\mathrm{out}} &:= C'_i | C'_{<i} = c_{<i}, I' \neq i \enspace. \\
\end{align*}
Because of the couplings between the permutations used in both executions, ${M_i}^{\mathrm{out}}$ and ${M'_i}^{\mathrm{out}}$ follow the same distribution since the element that differs between both datasets is not sent to the $i$th column under these conditionings.
On the other hand, ${M_i}^{\mathrm{in}}$ and ${M'_i}^{\mathrm{in}}$ are (again by the coupling condition) the results of applying a uniform shuffling to two databases of size $h$ differing in a single individual. Thus, by Theorem~\ref{thm:uniform-shuffling} we have the indistinguishability relation ${M_i}^{\mathrm{in}} \simeq_{\epsilon_S, \delta} {M'_i}^{\mathrm{in}}$ with $\epsilon_S := \epsilon_{\mathrm{clones}}(\epsilon_0, \delta, h)$.
Finally, leveraging the coupling used above to obtain \eqref{eqn:posterior} also obtain that $M_i^{\mathrm{in}}$ and $M_i^{\mathrm{out}}$ can be seen as two applications of a uniform shuffler to databases of size $h$ differing in a single individual. Therefore, we have ${M_i}^{\mathrm{in}} \simeq_{\epsilon_S, \delta} {M_i}^{\mathrm{out}}$ as well.

To obtain the privacy loss associated with the $i$th column we can now combine the facts above with Lemma~\ref{lem:sampling}.
In particular, using \eqref{eqn:posterior} and taking $p$ to be the distribution of $M_i^{\mathrm{in}}$, $q$ to be the distribution of ${M'_i}^{\mathrm{in}}$ and $r$ to be the distribution of ${M_i}^{\mathrm{out}}$, we obtain that $C_i | C_{<i} = c_{<i} \simeq_{\epsilon_C, \gamma \delta} C'_i | C'_{<i} = c_{<i}$ with
$\epsilon_C := \epsilon_{\mathrm{sampling}}(\epsilon_S, \gamma)$.

Finally, seeing the output of the mechanism as the adaptive composition of the $w$ different columns, we can use the advanced composition theorem \cite{DBLP:journals/tit/KairouzOV17} to obtain that $\alttwo$ applied to $R$ yields an $(\epsilon, w \gamma \delta + \delta')$-DP mechanism with
\begin{align*}
    \epsilon = \epsilon_C \left(\sqrt{2 w \log(1/\delta')} + k \frac{e^{\epsilon_C} - 1}{e^{\epsilon_C} + 1} \right) \enspace.
\end{align*}

Plugging in all the quantities derived above we obtain:
\begin{align*}
\epsilon
&=
    \epsilon_C \cdot \left(\sqrt{2 w \log(1/\delta')} + w \frac{e^{\epsilon_C} - 1}{e^{\epsilon_C} + 1} \right)
\\
&=
    \log\left(
    1
    +
    \gamma (e^{\epsilon_0}-1)
    \left(
    \sqrt{
    \frac{32 \log(4/\delta)}{(e^{\epsilon_0}+1) h}
    }
    +
    \frac{4}{h}
    \right)
    \right)
    \cdot
    \left(\sqrt{2 w \log(1/\delta')}
    +
    w \frac{\gamma (e^{\epsilon_0}-1)
    \left(
    \sqrt{
    \frac{32 \log(4/\delta)}{(e^{\epsilon_0}+1) h}
    }
    +
    \frac{4}{h}
    \right)}{2 + \gamma (e^{\epsilon_0}-1)
    \left(
    \sqrt{
    \frac{32 \log(4/\delta)}{(e^{\epsilon_0}+1) h}
    }
    +
    \frac{4}{h}
    \right)
    } \right)
\\
&=
    \log\left(
    1
    +
    \frac{e^{2\epsilon_0} (e^{\epsilon_0}-1)}{e^{2\epsilon_0} + w - 1}
    \left(
    \sqrt{
    \frac{32 \log(4/\delta)}{(e^{\epsilon_0}+1) h}
    }
    +
    \frac{4}{h}
    \right)
    \right)
    \cdot
    \left(\sqrt{2 w \log(1/\delta')}
    +
    w \frac{\frac{e^{2\epsilon_0} (e^{\epsilon_0}-1)}{e^{2\epsilon_0} + w - 1}
    \left(
    \sqrt{
    \frac{32 \log(4/\delta)}{(e^{\epsilon_0}+1) h}
    }
    +
    \frac{4}{h}
    \right)}{2 + \frac{e^{2\epsilon_0} (e^{\epsilon_0}-1)}{e^{2\epsilon_0} + w - 1}
    \left(
    \sqrt{
    \frac{32 \log(4/\delta)}{(e^{\epsilon_0}+1) h}
    }
    +
    \frac{4}{h}
    \right)
    } \right)
\\
&\leq
    \underbrace{
    \left(
    \frac{e^{2\epsilon_0} (e^{\epsilon_0}-1)}{e^{2\epsilon_0} + w - 1}
    \left(
    \sqrt{
    \frac{32 \log(4/\delta)}{(e^{\epsilon_0}+1) h}
    }
    +
    \frac{4}{h}
    \right)
    \right)
    }_{A}
    \cdot
    \underbrace{
    \left(\sqrt{2 w \log(1/\delta')}
    +
    w \frac{e^{2\epsilon_0} (e^{\epsilon_0}-1)}{e^{2\epsilon_0} + w - 1}
    \left(
    \sqrt{
    \frac{8 \log(4/\delta)}{(e^{\epsilon_0}+1) h}
    }
    +
    \frac{2}{h}
    \right)
    \right)
    }_{B}
\end{align*}
Ignoring the constants in these expressions, we see that
\begin{align*}
    A
    &=
    O\left(
    \frac{e^{3\epsilon_0} \sqrt{\log(1/\delta)}}{(e^{2\epsilon_0} + w) \min\{h, \sqrt{e^{\epsilon_0} h}\}}
    \right) \enspace,
    \\
    B
    &=
    O\left(
    \max\left\{
    \sqrt{w \log(1/\delta')},
    \frac{e^{3 \epsilon_0}}{e^{2\epsilon_0} + w} \max\left\{\frac{w}{h}, \sqrt{\frac{w^2 \log(1/\delta)}{h e^{\epsilon_0}}}\right\}
    \right\}
    \right) \enspace.
\end{align*}
In particular, assuming $k = h = w$ is large, $\epsilon_0 = O(1)$, and taking $\delta = \delta'$ we obtain
\begin{align*}
    \epsilon \leq 
    A \cdot B
    =
    O\left(
    \frac{e^{3\epsilon_0} \sqrt{\log(1/\delta)}}{(e^{2\epsilon_0} + k) \sqrt{e^{\epsilon_0} k}}
    \right)
    \cdot
    O\left(
    \sqrt{k \log(1/\delta)}
    \right)
    =
    O\left(
    \frac{e^{5\epsilon_0/2} \log(1/\delta)}{k}
    \right) \enspace.
\end{align*}
Recalling $k = \sqrt{n}$ completes the proof.

\subsection{Proof of Theorem~\ref{thm:weak-amplification-corrupted}}\label{app:proof-weak-amplification-corrupted}

Building on the proof of \Cref{thm:weak-amplification}, we can use a similar argument modified to exclude a small set of columns and rows.
In particular, we shall exclude columns that contain a corrupted user in the row where the target user is assigned: since the pre-protocol permutation is public, the corrupted users in this row can collude with the server to identify columns where the target user cannot be sent to, thus reducing the amount of amplification obtained from choosing a column at random.
Let us call the remaining columns \emph{valid}.
We shall also exclude the corrupted users from each valid column since these users can collude with the server to be excluded from the result of shuffling that column, reducing the amount of amplification obtained in that step.
Concentration bounds on the number and size of valid columns are easy to obtain as follows.
Recall that we denote the number of rows by $h$ and the number of columns by $w$.
The number of corrupted clients in the row where the target user is follows a Binomial distribution $\Bin(w-1, \gamma)$. Thus, from a tail bound we see that the row of the target user will, with probability at least $1 - \delta_w$, contain at least $w' = (1-\gamma)w - \log(1/\delta_w) - \sqrt{2 \gamma w \log(1/\delta_w)}$ honest users.
Furthermore, a tail and union bound show that, with probability at least $1 - \delta_h$, each column where an honest user occurs in the row of the target users contains at least $h' = (1-\gamma) h - \log(w/\delta_h) - \sqrt{2 \gamma h \log(w/\delta_h)}$ honest users.
The rest of the proof proceeds by replacing $w$ and $h$ with $w'$ and $h'$ in the proof of \Cref{thm:weak-amplification}, and adding a small probability of failure $\delta_w + \delta_h = 2\delta$ to the final privacy guarantee.

\subsection{Proof of Theorem~\ref{thm:no-strong}}

Let $n=k^2$ be the number of clients. Consider the neighboring databases $D_0,D_1$ with $x_1=0,1$ respectively, $x_2=...=x_k=1$ and $x_{k+1}=...=x_n=0$. Let $R(x)=1-x$ with probability $1/(1+e^{\epsilon_0})$ and $x$ otherwise.
Let $\vec{y}=\alttwo(R(\vec{x}))$.
Let $E$ be the event that $\exists j\leq k$ such that $y_j=y_{j+k}=...=y_{j+(k-1)k}=0$ i.e. there is an all zero column in the output.

For $E$ to happen on input $D_0$ it suffices for $R(x_1)=0$ and for the $k-1$ other inputs that are shuffled to the same column as $R(x_1)$ to also be mapped to zero. As all of these inputs are zero the probability of this happening is $(1-1/(1+e^{\epsilon_0}))^{-k}$ which is at least $1-k/(1+e^{\epsilon_0})$.

On the other hand, if $D_1$ is the input given then $E$ can't occur unless one of the first $k$ inputs is mapped to zero and as they are all one this probability is bounded (via union bound) by $k/(1+e^{\epsilon_0})$.

If $e^{\epsilon_0}$ is $\omega(k)$ then $P_{D_0}(E)=1-o(1)$ and $P_{D_1}(E)=o(1)$ and thus this mechanism isn't $(\epsilon,\delta)$-DP for any constant $\epsilon$ and $\delta$ bounded away from one.

\subsection{Proof of Theorem~\ref{theorem:ikos-all-honest-clients}}

Consider two arbitrary inputs $\vec{x}, \vec{x'}\in \G^n$ adding up to the same quantity. We  show that 
$\TV(\ikosview(\vec{x}), \ikosview(\vec{x'}))\leq 2^{-\sigma}$, for the value of $\sigma$ in the statement of the theorem.

Let $(\mat{M}^{(i)})_{i\in [m]}$ (resp. $(\mat{M'}^{(i)})_{i\in [m]}$) be $\sqrt{n}\times\sqrt{n}$ matrices corresponding to the state of the $m$ instances of
$\altshufflertwo{\pi, r_i}$ after applying the public permutation $\pi$, but {\em before shuffling}, to input $\tup{x}$ (resp. $\tup{x'}$).
Recall that all instances
$\altshufflertwo{\pi, r_i}$ use the same public randomness.
This means that all messages from user $j$
appear in the same row in matrices $\mat{M}^{(1)}, \ldots, \mat{M}^{(m)}$. We can thus think of the row shuffling step as
an application of IKOS. More concretely,
let $\mat{R}^{(1)}, \ldots, \mat{R}^{(m)}$
be the result of applying the row shuffling step
to each $\mat{M}^{(i)}$.
We can identify $\mat{R}^{(1)}, \ldots, \mat{R}^{(m)}$
with the result of $\sqrt{n}$ independent runs of IKOS with uniform shuffling, one for each row of $\mat{M}^{(1)} | \ldots | \mat{M}^{(m)}$, each involving $\sqrt{n}$ clients and $m$ messages. It then follows from Corollary 6.1 in~\cite{BalleBGN20} that
\begin{align*}
\TV&(\ikosview(\vec{x}), \ikosview(\vec{x'})) \leq 2^{-\sigma_1}\sqrt{n}~+\\ &\TV(\shuffleCols((\rv{S}^{(i)})_{i\in [m]}), \shuffleCols((\rv{S'}^{(i)})_{i\in [m]}))
\end{align*}
where $\sigma_1 := (m-1)(\log_2(n)/2 - \log_2(e)) - \log_2(q)$,
$\shuffleCols(.)$ applies independent random permutations to the columns of
the input matrices, 
and  $\rv{S}^{(i)}$ (resp. $\rv{S'}^{(i)}$) are uniformly random $\sqrt{n}\times\sqrt{n}$
matrices with elements in $\G$, conditioned on the sum of all the values in the $j$ row (across all $\mat{S}^{(i)}$'s) adding up to the same
amount than the corresponding row in the $\mat{M}^{(i)}$'s (resp. $\mat{M'}^{(i)}$'s), i.e.
$\forall j\in[\sqrt{n}]: \sum_{i,k} \rv{S}^{(i)}_{j, k} = \sum_{i,k} \mat{M}^{(i)}_{j, k}$. Note that the entries in $(\rv{S}^{(i)})_{i\in[m]}$ and $(\rv{S'}^{(i)})_{i\in[m]}$ add up to the same value by definition of $\tup{x}, \tup{x'}$.
Let us now consider the columns of $\rv{S}^{(i)}, \rv{S'}^{(i)}$. We can identify the application of $\shuffleCols$ with a run of IKOS (with uniform shuffling) involving $\sqrt{n}$ clients and $m\sqrt{n}$ messages per client. Then by Corollary 6.1 in~\cite{BalleBGN20} 
we have 
\begin{align*}
\TV&(\ikosview(\vec{x}), \ikosview(\vec{x'})) \leq 2^{-\sigma_1}\sqrt{n}~+\\
&2^{-\sigma_2}\leq 2^{-\sigma_1}(1+\sqrt{n}) < 2^{-\sigma_1+\lceil\log_2(1+\sqrt{n})\rceil} < 2^{-\sigma}
\end{align*}
with 
$\sigma_2 := (m\sqrt{n}-1)(\log_2(n)/2 - \log_2(e)) - \log_2(q)$ and $\sigma$ defined as in the statement of the theorem.

\subsection{Proof of Theorem~\ref{thm:ikos-corrupted}}

Let $(\mat{M}^{(i)})_{i\in [m]}$ (resp. $(\mat{M'}^{(i)})_{i\in [m]}$) be $\sqrt{n}\times\sqrt{n}$ matrices corresponding to the state of the $m$ instances of
$\altshufflertwo{\pi, r_i}$, after applying the public permutation $\pi$, but {\em before shuffling}, to input $\tup{x}$ (resp. $\tup{x'}$).
For each $\mat{M}^{(i)}$, let $\rv{X_j}$ be
the number of 
honest clients whose input is placed in 
row $j$.
Let us call $\alpha$ the fraction of honest clients in the population
and note that $\rv{X_j}$ is a
hypergeometric
random variable with parameters 
$n, \alpha n, \sqrt{n}$
for population size (number of clients),
number of ``marked'' elements (number of honest clients),
and sample size (number of clients contributing to row $j$), respectively.
Note that by the assumption of the theorem $\alpha > (1-\gamma)n$
and thus we can establish a high probability lower bound $\ell$ on 
$\max_{j} \rv{X_j}$. More concretely, it follows from a tail bound
on $\rv{X_j}$ and a union bound over all $j\in [n]$ that
$\max_{j} \rv{X_j} \geq \ell$ for $\ell = (1-\gamma)\sqrt{n} - (\sigma_1 + \log n)^{1/2}n^{1/4}$,
except with probability less than $2^{-\sigma_1}$. Note that this probability is over the public randomness
$\pi$ used to assign clients rows of $\mat{M}^{(i)}$, and common to all shufflers.
We can thus assume that every every row shuffle in the execution of 
$m$ instances of $\altshufflertwo{\pi, r_i}$ involves at least $\ell$ shares from honest 
users (and thus unknown to the server).
We can then consider the $\ell\times\ell$ submatrices of $(\mat{M}^{(i)})_{i\in [m]}$ (resp. $(\mat{M'}^{(i)})_{i\in [m]}$) involving the $\ell$ honest users per row only, and repeat the argument of the proof of Theorem~\ref{theorem:ikos-all-honest-clients}.

\subsection{Proof of Theorem~\ref{thm:no-do}}

Consider the possibility that all $n$ inputs are different. After the public shuffle we can assume wlog that one of the two inputs to be swapped is in the first position.
If the second is in the first $k$ then they will be shuffled next giving perfect indistinguishability, therefore we concentrate on the probability $n(n-1)/(n^2-1)$ event that it is not.
We can then assume that the second swapped input is in the $k+1$th position wlog.
We are then left to show that applying $\alttwo$ to each of
\begin{equation*}
\begin{pmatrix}
1 & 2 & \cdots & k \\
k+1 & k+2 &\cdots & 2k \\
\vdots & & &\vdots \\
\end{pmatrix}
\textrm{ and }
\begin{pmatrix}
k+1 & 2 & \cdots & k \\
1 & k+2 & \cdots & 2k \\
\vdots & & & \vdots \\
\end{pmatrix}
\end{equation*}
doesn't render then indistinguishable.

Note that the entry in the top left can't end up in the same column as $2,...,k$ which must also each be in a separate column.
Thus by looking at the locations of $2,...,k$ the adversary can work out which column the the top left entry ended up in.
If whichever of $1$ and $k+1$ weren't in the top left corner is not also in that column, then the adversary can tell which one that is i.e. the adversary can distinguish the two cases with certainty.

The probability of $1$ and $k+1$ not ending up in the same column is $n-1/n$.
Multiplying this by the probability of the earlier event we conditioned on gives $(n-1)^2/(n^2-1)=(n-1)/(n+1)$ as the adversaries probability of distinguishing with certainty.
The conclusion is immediate.

\section{Client Runtime Benchmarks}
\label{sec:appendix-runtimes}

Table~\ref{tab:exponentiations-ms} is analogous to  Table~\ref{tab:exponentiations}, but reporting time in milliseconds,
for both a Pixel7 device and a standard laptop.
As the underlying group we used the (standard) elliptic curve Curve25519, and benchmarks were done with the Dalek-Cryptography framework~\cite{daleklib}, written in Rust. We benchmarked on both a standard laptop, and a Pixel7 device.
A single exponentiation in a Pixel7 phone and a standard laptop take $0.045$ms and $0.0115$ms, respectively.
For amortized shuffling we report timing using Pippenger's algorithm.

\begin{table*}
    \centering
    \begin{tabular}{c|cc|cc}
        Number of Clients & \multicolumn{2}{c}{$10^3$} & \multicolumn{2}{c}{$10^5$}\\ \hline
        Key Agreement & \multicolumn{2}{c}{1.28/0.32} & \multicolumn{2}{c}{1.47/0.37} \\
        Decryption & \multicolumn{2}{c}{1.79/0.45} & \multicolumn{2}{c}{1.97/0.50} \\
        & \textsc{Avg} & \textsc{Worst} & \textsc{Avg} & \textsc{Worst} \\\cline{2-5}
        Shuffling (Amortized) & 4.83/1.22 & 78.2/18.86 & 4.96/1.24 & 5230/3077 \\
        Shuffling (Alternating) & 24.97/6.27 & 15.18/3.81 & 12.13/3.05 & 95.74/24.07
    \end{tabular}
    \vspace{1cm}
    \caption{\small Runtime (in ms.) of exponentiations per client in each part of the protocol. All values are for $\sigma=40$, $\gamma =1/20$, $\alpha=1/20$, $\eta= 10$. Shuffling rows display counts averaged over all clients and for the client with the worst load. Entries of the form X/Y corresponds to measurements in a Pixel7 device and a standard laptop. For Shuffling(Alternating) - worst we benchmark the accelerated variant via efficient multi-exponentations.}
    \label{tab:exponentiations-ms}
\end{table*}

\section{Proofs for Section~\ref{sec:implementations}}\label{sec:appendix-implementations}

\begin{proof}[Proof of Theorem~\ref{thm:sharedkeyprotocols}]
We will start with the properties for Figure~\ref{Protocol:KeyAgreement}.
The commitments sent by $c_{i,j}$ in step 3, even if they are group elements generated in an arbitrary malicious way, are consistent with a unique $P_{i,j}$ and $Q_{i,j}$ that have the same constant term which we will call $s_{i,j}$ (as that is what it will be for the honest clients). The shares that $c_{i,j}$ sends to the honest clients must be evaluations of this $P_{i,j}$ and $Q_{i,j}$ at the appropriate point, otherwise they will be reported in step 4 and dropped in step 5 (along with all their remaining influence on the protocol).

Therefore at step 6 each honest client will hold a share of $s_{i,j}$ for each $c_{i,j}$ who has not been dropped. By linearity of Shamir secret sharing these are added to a sharing of one secret for that committee. This is uniformly random from the adversaries perspective so long as at least one honest clients submitted shares that were accepted and no more than $t-1$ of these new shares are held by malicious clients (within each committee).

Each committee now has an independent secret shared amongst them, but crucially that same secret shared amongst the next committee. Learning the difference between each the secrets of consecutive committees leaks nothing about the first committees secret. However it does allow the server to tell each committee the offset of their secret from party ones secret. Any dishonest reporting by malicious client in step 6 is caught in step 7. All the honest parties then subtract their offset so they have shares of the first committees secret and output them as required. The malicious parties can output whatever they want.

The only way this Key Agreement protocol could fail (with a semi-honest server) is if the server didn't receive enough honest responses to reconstruct the offsets. This won't happen so long as at least $t$ clients follow the protocol (without dropping out).

Turning to Figure~\ref{Protocol:Decryption}. We can be confident that the clients don't learn anything other than the leakage because the only other data they receive are the uniformly random challenges.
It remains to show that the checks in step 5 constitute a check of a valid proof that the $v_{i,j}$ have been validly constructed and that this proof is zero-knowledge. 

It is zero knowledge because the only (non-leakage) input dependant message sent to $\server$ is $e_j$ in step 4 and this message is masked by $r_j$ which is still pseudorandom from the perspective of $\server$ by the decisional Diffie-Hellman assumption even after being given all the other messages about it i.e. $A_j$ and the $B_{i,j}$.

To see that it is a proof, the server can define $r_j$ to be the logarithm in the base $g$ of $A_j$ (which for an honest client it would be), then the only way to pass step 5a is for the client to set $e_j=r_j+u_j\sk_{a,j}$. For each fixed value of $i$ the client must (in order to pass 5b) have submitted a $B_{i,j}$ and a $v_{i,j}$ satisfying the condition in 5b, but they must be submitted before seeing $u_j$ and any values other than the honest ones will only satisfy this equality for one possible value of $u_j$. The probability of the server choosing that $u_j$ is negligible and we are done.
\end{proof}

\begin{proof}[Proof of Theorem~\ref{thm:amortizedsecurity}]
As $\server$ is semi-honest and our three building blocks enforce that all client behaviour is verified except the submission of inputs we can assume the protocol is followed apart from when the clients are providing input. We must verify two things: firstly that the conditions for the building blocks hold, i.e. there are no more than $t-1$ clients in each decryption committee, and secondly that at least one of the shufflers is honest, so that the permutation is uniformly random. The leakage of the decryption protocol is not a problem as the $h_i$ sent to the clients are potentially known to the adversary anyway (they could have been sent by a malicious party in the previous round) and the values of $h_i^{\sk_{a,j}}$ are a one way function of the secret keys.

The number of corrupt clients in the first committee is at most a hypergeometric random variable with parameters $n,\committeesize$ and $\gamma n$. A tail-bound gives that this will be at most $t-1$ with probability $1-\exp(-2(t/\committeesize-\gamma)^2\committeesize)$. The probability of any committee having too many bad clients is, by a union bound, at most $m\exp(-2(t/\committeesize-\gamma)^2\committeesize)$, thus at most $2^{-\sigma-1}$.

In order to have one honest non-dropped shuffler we require that there are at most $\shufflecommitteesize-d-1$ malicious clients amongst the selected $\shufflecommitteesize$. That is that a hypergeometric variable with parameters $n,\shufflecommitteesize$ and $\gamma n$ is not more than $\shufflecommitteesize-d-1$. A tail-bound gives that this bad event has probability at most $\exp(-2(1-d/\shufflecommitteesize-\gamma)^2\shufflecommitteesize)$, thus at most $2^{-\sigma-1}$.

By a union bound we get the required $2^{-\sigma}$ probability of the bad event.

The argument bounding the probability of an abort is similar. We merely require that each decryption committee has at least $t$ protocol following clients and at least $\shufflecommitteesize-d$ protocol following clients are selected to shuffle. The probability of the former failing to be true is bounded by $m\exp(-2((t+1)/\committeesize-(1-\alpha))^2\committeesize)$ and the latter is bounded by $\exp(-2((d+1)/\shufflecommitteesize-delta)^2\shufflecommitteesize)$, giving the required bounds.
\end{proof}

\begin{proof}[Proof of Theorem~\ref{thm:Alternating}]
This proof is much the same as the proof of Theorem~\ref{thm:amortizedsecurity}. The decryption committees are the same as in that proof so we again have the $-\log_2(m)+2log_2(e)(t/\committeesize-\gamma)^2\committeesize$ term in $\sigma$. Regarding the shuffling committees, as a worst case we can assume a separate committee is drawn for each run of the ciphertext shuffle. There are $h\lceil l/2 \rceil + w \lfloor l/2 \rfloor$ such committees so we need only subtract the logarithm of this from the expression from the amortized case.

Similarly the correctness parameter $\eta$ is the same as in the amortized case except for the subtraction of a $\log_2(h\lceil l/2 \rceil + w \lfloor l/2 \rfloor)$ term from the shuffle committee component.
\end{proof}

\end{document}